\documentclass[preprint,nocomments,norevision]{ptephy_om}
\preprintnumber{RIKEN-iTHEMS-Report-26}
\revisionnum{1}

\newcommand{\parity}{\mathsf P}
\newcommand{\varEP}[1]{\vsub{#1}{EP}}
\newcommand{\Laurent}{\mathcal L}

\newtheorem{theorem}{Theorem}[section]
\newtheorem{proposition}{Proposition}[section]

\usepackage{hyperref}
\usepackage{orcidlink}

\title{Riesz--Laurent representation of black-hole scattering and sourced response at exceptional points}

\author[1]{Okuto Morikawa \orcidlink{0000-0002-0044-4491}}
\affil[1]{Center for Interdisciplinary Theoretical and Mathematical Sciences (iTHEMS),
RIKEN, Wako 351-0198, Japan
\email{okuto.morikawa@riken.jp}
}

\author[2]{Shoya Ogawa \orcidlink{0000-0003-0900-2486}}
\affil[2]{Department of Physics, Kyushu University, 744 Motooka, Nishi-ku,
Fukuoka 819-0395, Japan
\email{ogawa.shoya.615@m.kyushu-u.ac.jp}
}

\author[3]{Takuya Hirose \orcidlink{0000-0003-0962-8884}}
\affil[3]{Faculty of Science and Engineering, Kyushu Sangyo University,
Fukuoka 813-8503, Japan
\email{t.hirose@ip.kyusan-u.ac.jp}
}

\begin{document}
\begin{abstract}
At a black-hole exceptional point (EP), two quasinormal modes coalesce
and their separate residues become ill-conditioned.  Rather than postulating
a near-degenerate modal fit, we derive the response constructively from the
complex-scaled Regge--Wheeler--Zerilli resolvent, treating the modes as one
isolated rank-two Riesz cluster.  Its zeroth and first contour moments
determine an exact pair resolvent on both sides of, and at, the EP, without
labeling the individual modes or constructing a normalized Jordan chain.  At
a second-order EP, these moments determine the simple- and double-pole Laurent
operators.  Although the modal decomposition is singular, fixed-real-frequency
transmission and the greybody factor remain real-analytic through the EP,
provided that the cluster remains isolated, the complementary resolvent is
regular, and no pole reaches the physical axis.  Source--observer matrix
elements of the Laurent operators define finite, normalization-independent
amplitudes and fix both the constant and linear-in-time terms in the causal
ringdown.  Their equality with the coefficients from the Jost double-zero
expansion shows that they are operator-defined coefficients of the specified
physical response, rather than fitting parameters.  Thus two cluster moments
provide mode-label-free data from which both scattering and driven responses
follow.
\end{abstract}
\maketitle

\tableofcontents

\section{Introduction}
\label{sec:introduction}

\subsection{Resolvent viewpoint for black-hole scattering}
Linear perturbations of a stationary black hole constitute an open scattering
problem.  After decomposition into angular modes and Fourier transformation in
time, scalar, electromagnetic, and gravitational perturbations are frequently
reduced to a one-dimensional equation,
\begin{align}
    \left(z-H_{\ell}\right)\psi_{\ell\omega}(x)
    &=
    S_{\ell\omega}(x),
    \label{eq:intro_driven_equation}
    \\
    H_{\ell}
    &=
    -\frac{\rmd^2}{\rmd x^2}+V_{\ell}(x),
    \qquad
    z=\omega^2,
    \label{eq:intro_radial_operator}
\end{align}
where $x=r_*$ is the tortoise coordinate\footnote{%
For the automated construction of tortoise-coordinate maps for static,
spherically symmetric black-hole metrics, see the Julia package
\texttt{AutoTortoise.jl}~\cite{Morikawa:2026AutoTortoise}.
The package identifies the relevant horizons, constructs the corresponding
tortoise map, and provides a numerical inverse map $r=r(r_*)$, which is
useful for expressing black-hole effective potentials directly as functions
of the tortoise coordinate.
}, $V_{\ell}$ is an effective radial
potential, and $S_{\ell\omega}$ is a possible external source.  The object that
simultaneously controls scattering and driven response is the outgoing
resolvent,
\begin{equation}
    R_{\ell}(z)
    =
    \left(z-H_{\ell}\right)^{-1},
    \qquad
    \psi_{\ell\omega}
    =
    R_{\ell}(z)S_{\ell\omega},
    \label{eq:intro_resolvent}
\end{equation}
understood as the retarded boundary value on the physical sheet and as its
analytic continuation when complex frequencies are considered.

Equation~\eqref{eq:intro_resolvent} may be read as an input--output map.
The operator $H_{\ell}$ together with the radiation conditions specifies the
black-hole scattering problem, the source $S_{\ell\omega}$ specifies how that
problem is driven, and a subsequent matrix element specifies how the response
is extracted.  The pole locations are properties of the first ingredient,
whereas an observed amplitude depends on all three.  This separation is useful
throughout the paper because changing the source or observation channel does
not move the QNM spectrum, although it can strongly enhance or suppress the
contribution of a given spectral sector.

Quasinormal modes (QNMs) are poles of this analytically continued resolvent, or
equivalently of the corresponding outgoing Green function
\cite{Leaver:1986gd,Berti:2009kk,Konoplya:2011qq}.  Their frequencies describe
the characteristic oscillation and decay scales of the black hole.  The full
resolvent contains more information than the pole locations, however.  In an
asymptotically flat spacetime its non-pole sector is required for the prompt
response and the branch-cut contribution that produces late-time tails
\cite{Ching:1994bd,Casals:2013mpa}.  On the real-frequency axis, the same
operator determines the reflection and transmission amplitudes.  For a wave
incident from spatial infinity, we use the convention
\begin{align}
    \psi_{\ell\omega}(x)
    &\sim
    \rme^{-\rmi\omega x}
    +\mathcal R_{\ell}(\omega)\rme^{+\rmi\omega x},
    &&x\rightarrow+\infty,
    \label{eq:intro_scattering_infinity}
    \\
    \psi_{\ell\omega}(x)
    &\sim
    \mathcal T_{\ell}(\omega)\rme^{-\rmi\omega x},
    &&x\rightarrow-\infty.
    \label{eq:intro_scattering_horizon}
\end{align}
For the real Schwarzschild potential and the usual unit-flux normalization,
the greybody factor is
\begin{equation}
    \Gamma_{\ell}(\omega)
    =
    \frac{\vsub{\mathcal F}{H}}{\vsub{\mathcal F}{in}}
    =
    \left|\mathcal T_{\ell}(\omega)\right|^2.
    \label{eq:intro_greybody}
\end{equation}
For a recent proposal to model ringdown spectra directly through
greybody factors, see Ref.~\cite{Oshita:2023cjz}.
The QNM spectrum, the continuum contribution, and the greybody factor are thus
different aspects of the same radial scattering problem.

The resolvent is equally central when an inhomogeneous source is present.  A
localized impulse, an infalling compact object, or an extended matter
distribution produces a field through $R_{\ell}S_{\ell\omega}$.  Near a simple
QNM pole, the residue factorizes into a source overlap, a propagation factor,
and a pole denominator.  Consequently, a QNM frequency is a property of the
operator and its radiation conditions, whereas its observable amplitude also
depends on how the system is driven and where the response is measured.
Green-function studies of particle-driven perturbations have made this source
and history dependence explicit \cite{DeAmicis:2025dynamical}.  A framework
that retains the complete resolvent can therefore treat QNM excitation without
discarding the accompanying nonresonant response.

A number of methods accurately address particular parts of this structure.
Continued fractions and analytic matching are highly effective for QNM
frequencies~\cite{Leaver:1985ax,Leaver:1986gd}, while direct Jost integration and time-domain evolution provide
scattering amplitudes and retarded waveforms.  Hyperboloidal formulations
regularize the event horizon and future null infinity geometrically~\cite{Ansorg:2016ztf,Jaramillo:2020tuu} and, when
combined with Keldysh theory, lead to resonant expansions of a non-self-adjoint
time-evolution generator
\cite{PanossoMacedo:2024nkw,Besson:2024keldysh}. See also Refs.~\cite{Miyachi:2025ptm,Pombo:2025urp}.  The present work takes a
complementary stationary viewpoint: the primary spectral object is the
frequency-domain two-channel scattering resolvent itself.  This choice is
motivated by the fact that its kernel and channel matrix elements give access
to transmission and source response in addition to resonant frequencies.

The complex scaling method (CSM) provides a direct spectral representation of
this resolvent.  An analytic deformation of the asymptotic radial contour maps
outgoing waves to square-integrable functions and rotates the continuous
spectrum away from the physical axis.  Resonances exposed between the physical
axis and the rotated cut become isolated eigenvalues of a non-Hermitian
operator \cite{Aguilar:1971ve,Balslev:1971vb,Simon1979ECS,Moiseyev:1998gjp,
Myo:2014ypa}.  In an extended biorthogonal resolution, the deformed resolvent
has the schematic form
\begin{align}
    R_{\ell,\theta}(z)
    & =
    \sum_{n\in\mathrm{res}}
    \frac{
        \ket{\psi_{n,\theta}^{R}}
        \bra{\psi_{n,\theta}^{L}}
    }{z-E_n}
    +
    \int_{L_{\theta}}\rmd\lambda\,
    \frac{
        \ket{\psi_{\lambda,\theta}^{R}}
        \bra{\psi_{\lambda,\theta}^{L}}
    }{z-\lambda},
    \label{eq:intro_spectral_resolvent}
\end{align}
where $L_{\theta}$ denotes the rotated continuum.  A finite $L^2$ basis
replaces the continuum integral by a set of pseudostates.  The pole and
continuum sectors can then be evaluated from the same matrix inverse rather
than by combining unrelated numerical constructions.

Our previous work used complex scaling to identify QNMs of Schwarzschild and
Reissner--Nordstr\"om black holes as isolated complex eigenvalues
\cite{Ogawa:2026veu}.  The method was subsequently extended to the continuum
level density, which supplies trace-level information about the spectral shift
and the total scattering phase \cite{Ogawa:2026opj}.  A trace observable does
not, in general, determine an off-diagonal channel amplitude or its modulus.
In particular, reconstructing $\Gamma_{\ell}=|\mathcal T_{\ell}|^2$ requires
matrix elements of the Green kernel rather than only the eigenvalue locations
or the resolvent trace.  This motivates the transition from a complex-scaled
QNM calculation to a complex-scaled scattering-resolvent calculation.

\subsection{Exceptional points and the failure of modal coordinates}
A second motivation is provided by exceptional points (EPs), at which two QNM
eigenvalues and their eigenvectors coalesce~\cite{Moiseyev:2011,Rotter:2007zng,Heiss:2012dx}. A concrete black-hole realization was identified in the massive-scalar
QNM spectrum of near-extremal Kerr black holes, where the longest-lived
mode and the first overtone coalesce and exhibit the characteristic
hysteretic mode exchange under adiabatic encircling
\cite{Cavalcante:2024swt}.
A companion near-extremal analysis mapped the surrounding QNM spectrum
and provided further numerical evidence for the same EP
\cite{Cavalcante:2024kmy}. Close to an EP, assigning a
separate excitation amplitude to each of the two modes becomes ill-conditioned,
although their combined response remains finite.  Recent studies have related
avoided crossings, resonant enhancement, and polynomially modulated ringdown
to this non-Hermitian structure
\cite{Motohashi:2024fwt,Yang:2025dbn,PanossoMacedo:2025xnf}.  The resolvent is
the natural object at the coalescence itself: a second-order EP is represented
by a double pole and a corresponding Jordan chain rather than by two ordinary
simple-pole projectors.

Near the EP, a signal may be represented phenomenologically by two nearly
degenerate damped exponentials or, at the coalescence, by a term of the form
$(a+bt)\rme^{-\rmi\varEP{\omega}t}$.  Such ans\"atze are useful for waveform
analysis, while relating their fitted coefficients to the underlying spectral
structure generally requires additional spectral input.  The fitted
coefficients combine properties of the operator with the source, the
observation channel, and the fitting prescription.  Our question precedes
such modeling.  We ask which quantities are defined constructively by the
non-Hermitian operator and its resolvent, and how a physical waveform follows
from them.  In this viewpoint, the polynomially modulated ringdown is a
consequence of the Laurent structure of the resolvent, rather than an assumed
template.

Complementary semi-analytic work in the Nariai limit has analyzed the
excitation of nearly double-pole QNMs, including destructive interference
and transient linear-in-time growth near exceptional lines
\cite{Nakamoto:2026lyo}.
Recent time-domain calculations for scalar perturbations of hairy black holes
have shown directly that an EP ansatz containing a term linear in time can
describe the coalescent ringdown more robustly than a sum of independent
damped modes~\cite{Cheng:2026ep}.
These results demonstrate the physical relevance of the EP waveform form.
The complementary question addressed here is how the coefficients of a
specified source--observer response are obtained directly as matrix
elements of finite Laurent operators of the full resolvent, without
fitting a waveform or assigning separately normalized amplitudes to the
coalescing QNMs.

Riesz projections and nondivergent representations at an EP are established
tools of non-Hermitian spectral theory
\cite{Kato:1976,Hashimoto:2014ep}.\footnote{%
Riesz projections provide a basis-independent construction of isolated
spectral subspaces~\cite{Kato:1976}.  A complementary nondivergent
description was developed by Hashimoto \textit{et al.\/}~\cite{Hashimoto:2014ep}, who replaced the singular eigenvector expansion
near an EP by an extended pseudo-eigenstate basis that approaches a Jordan
chain continuously.  Once the rank-two Riesz subspace has been selected,
their generalized Jordan representation and the present pair of moments
$\vsub{P}{pair}^{\theta}$ and
$\vsub{\mathcal M}{pair}^{\theta}$ describe the same restricted operator.
The distinction is that the Riesz moments construct both the invariant
subspace and the operator acting on it directly from the full black-hole
resolvent, without choosing an extended pseudo-eigenstate basis.}  Our contribution is not to reintroduce
these tools in isolation, but to use the first two contour moments to construct
the black-hole pair resolvent and then derive its scattering and
source--observer consequences within one stationary framework.

\subsection{Riesz-cluster construction and main results}
It is helpful to keep the roles of the variables distinct.  The complex number
$z$ is the spectral argument at which the resolvent is evaluated.  The real
parameters $q$ change the operator and are used to move through the EP.  The
angle $\theta$ changes the complex-scaled representation but not the underlying
physical matrix element in the exact theory.  Finally, the source $f$ and the
observation functional $g$ select an input--output channel without changing the
operator spectrum.  Thus a pole in $z$, a branch structure in $q$, and a large
or small channel response are related but logically different statements.

The need to avoid individual mode labels is not merely numerical.
In the Kerr EP cascade, the overtone ordering at any fixed parameter point
is unambiguous, but adiabatic continuation around the EPs can permute the
modes, so that a mode identity continued through parameter space is
path dependent
\cite{Cavalcante:2025abr}.
For the local problem considered here, this motivates treating the selected
resonant pair as one invariant spectral cluster rather than choosing a
globally continued label for either member.

The central object is therefore not either member of the coalescing QNM pair
separately, but the isolated rank-two root subspace containing both modes.
Let $q$ denote the operator parameters, and let
$\vsub{P}{pair}^{\theta}(q)$ and $\vsub{\mathcal M}{pair}^{\theta}(q)$ be the zeroth
and first Riesz moments of the complex-scaled resolvent around this cluster.
They determine
\begin{align}
    \vsub{E}{c}(q)
    &=
    \frac{1}{2}\Tr\vsub{\mathcal M}{pair}^{\theta}(q),
    \\
    \vsub{K}{pair}^{\theta}(q)
    &=
    \vsub{\mathcal M}{pair}^{\theta}(q)
    -\vsub{E}{c}(q)\vsub{P}{pair}^{\theta}(q),
    \\
    \Delta(q)^2
    &=
    \frac{1}{2}\Tr
    \left[
        \vsub{K}{pair}^{\theta}(q)^2
    \right].
\end{align}
Our principal analytic result is the exact identity
\begin{equation}
    \vsub{R}{pair}^{\theta}(z;q)
    =
    \frac{
        [z-\vsub{E}{c}(q)]
        \vsub{P}{pair}^{\theta}(q)
        +\vsub{K}{pair}^{\theta}(q)
    }{
        [z-\vsub{E}{c}(q)]^2
        -\Delta(q)^2
    }.
    \label{eq:intro_exact_pair_resolvent}
\end{equation}
The individual eigenvalues and projectors may exchange sheets or become
ill-conditioned, whereas the quantities on the right-hand side are
single-valued as long as the contour continues to isolate the same cluster.
At a second-order EP, $\vsub{P}{pair}^{\theta}$ and
$\vsub{K}{pair}^{\theta}$ become the simple- and double-pole Laurent
operators, respectively.  Thus the regular data are the invariant subspace
and the restriction of the operator to it, rather than amplitudes assigned to
two singularly defined basis vectors.

This operator identity has two physical consequences.  First, its scattering
channel matrix elements give an exact pair/rest decomposition of the
transmission amplitude.  Although the QNM decomposition is singular at the
EP, the fixed-real-frequency transmission amplitude and greybody factor remain
real-analytic, provided that no pole reaches the physical axis and the
complementary resolvent remains regular.  Second, matrix elements between a
prescribed source and observation functional define finite generalized
amplitudes at the coalescence.  Their inverse Fourier transform fixes both the
constant and linear-in-time terms multiplying the degenerate ringdown
frequency.  We show that these coefficients agree with the double-zero
expansion of the Jost Wronskian.  They are therefore spectral invariants of a
specified source--observer channel, not free parameters introduced by a
waveform fit.  For dilation-analytic states, the same quantities are
independent of the complex-scaling angle in the exact theory.

We deliberately separate this analytic construction from its numerical
realization.  A finite-basis formulation is retained because it provides a
direct implementation of the contour moments without resolving nearly
parallel eigenvectors, and the two-parameter Gaussian deformation of the
Regge--Wheeler potential in Ref.~\cite{Yang:2025dbn} supplies a concrete
black-hole setting.  The deformation changes the spectral operator and is
distinct from the inhomogeneous source, so the creation of the EP and the
coupling of a chosen channel to its double-pole sector can be addressed
separately.  Numerical calculations can test the contour construction and
quantify the response, but they do not define the cluster invariants derived
here.  Their quantitative implementation is left to subsequent work.

The remainder of the paper is organized as follows.  In
Sec.~\ref{sec:scattering_setup}, we define the sourced Schwarzschild
perturbation problem and its outgoing Green function.  Section~\ref{sec:complex_scaling}
constructs the global and exterior complex-scaled representations and their
finite-basis realization.  In Sec.~\ref{sec:real_axis_ep}, we derive the exact
rank-two Riesz-cluster formula and apply it to pair-resolved
real-frequency scattering.  Section~\ref{sec:source_response} introduces the
localized impulse and its QNM and continuum response.  The
Gaussian-deformed Regge--Wheeler model, the Jordan interpretation of the
cluster moments, the EP Laurent expansion, and the generalized source
amplitudes are discussed in Sec.~\ref{sec:exceptional_points}.
Section~\ref{sec:conclusion} summarizes the analytic framework and its scope.
Appendix~\ref{app:numerical_realization} gives a concrete numerical workflow
for constructing the two contour moments with the accompanying Julia tools.

\section{Scattering setup}
\label{sec:scattering_setup}

\subsection{Regge--Wheeler--Zerilli equations with a source}
\label{sec:rwz_source}

We consider linear perturbations of a Schwarzschild black hole of mass
$M$.  The background metric is
\begin{equation}
    \rmd s^2
    =
    -F(r)\,\rmd t^2
    +\frac{\rmd r^2}{F(r)}
    +r^2\rmd\Omega^2,
    \qquad
    F(r)=1-\frac{2M}{r}.
\end{equation}
The tortoise coordinate $x:=r_*$ is defined by
\begin{equation}
    \frac{\rmd x}{\rmd r}
    =
    \frac{1}{F(r)},
    \qquad
    x
    =
    r+2M\ln\left(\frac{r}{2M}-1\right).
    \label{eq:tortoise_coordinate}
\end{equation}
Thus, the event horizon and spatial infinity correspond to
$x\rightarrow-\infty$ and $x\rightarrow+\infty$, respectively.

For each radiative multipole $\ell\geq2$, the gauge-invariant
Zerilli--Moncrief and Cunningham--Price--Moncrief variables satisfy a
one-dimensional wave equation.  We denote the parity by
$\parity\in\{\mathrm{odd},\mathrm{even}\}$ and write
\begin{equation}
    \left[
        -\frac{\partial^2}{\partial t^2}
        +\frac{\partial^2}{\partial x^2}
        -V_{\ell}^{\parity}(r)
    \right]
    \psi_{\ell m}^{\parity}(t,x)
    =
    S_{\ell m}^{\parity}(t,x).
    \label{eq:rwz_time_domain}
\end{equation}
For gravitational perturbations, the odd-parity Regge--Wheeler
potential \cite{Regge:1957td} is
\begin{equation}
    \vsup{V}{odd}_{\ell}(r)
    =
    F(r)
    \left[
        \frac{\ell(\ell+1)}{r^2}
        -\frac{6M}{r^3}
    \right],
    \label{eq:rw_potential}
\end{equation}
whereas the even-parity Zerilli potential \cite{Zerilli:1970se} is
\begin{align}
    \vsup{V}{even}_{\ell}(r)
    &=
    \frac{2F(r)\,\mathcal Z_{\ell}(r)}
    {r^3\left(\lambda r+3M\right)^2},
    \label{eq:zerilli_potential}
    \\
    \mathcal Z_{\ell}(r)
    &=
    \lambda^2(\lambda+1)r^3
    +3\lambda^2Mr^2
    +9\lambda M^2r
    +9M^3,
    \\
    \lambda
    &=
    \frac{(\ell-1)(\ell+2)}{2}.
\end{align}
The source $S_{\ell m}^{\parity}$ is obtained by projecting the perturbed
stress-energy tensor onto the corresponding tensor harmonics.  It
vanishes for a vacuum perturbation, while a point particle generally
produces terms proportional to both a radial delta function and its
derivative.  In the present work we first use smooth localized sources,
so that the properties of the resolvent can be separated from the
additional distributional issues associated with a particle worldline.
The general relation between stress-energy projections, master
variables, and radiative observables is reviewed in
Ref.~\cite{Martel:2005ir}.

Once the master variables are known, the asymptotic metric perturbation
and the associated gravitational-wave flux can be reconstructed using the
standard Regge--Wheeler--Zerilli formulas; see, e.g.,
Ref.~\cite{Martel:2005ir}.\footnote{%
In the
Zerilli--Moncrief and Cunningham--Price--Moncrief normalization used by
Ref.~\cite{Martel:2005ir}, one convenient convention is
\begin{align}
    h_+-\rmi h_\times
    & =
    \frac{1}{2r}
    \sum_{\ell, m}
    \sqrt{\frac{(\ell+2)!}{(\ell-2)!}}
    \left(
        \vsup{\psi}{even}_{\ell m}
        +\rmi\vsup{\psi}{odd}_{\ell m}
    \right)
    {}_{-2}Y_{\ell m}
    +O(r^{-2}),
    \label{eq:rwz_strain}
\end{align}
and the time-averaged energy flux at future null infinity is
\begin{align}
    \Braket{\frac{\rmd E}{\rmd u}}
    &=
    \frac{1}{64\pi}
    \sum_{\ell, m}
    \frac{(\ell+2)!}{(\ell-2)!}
    \Braket{
        \left|\partial_u\vsup{\psi}{even}_{\ell m}\right|^2
        +
        \left|\partial_u\vsup{\psi}{odd}_{\ell m}\right|^2
    }.
    \label{eq:rwz_energy_flux}
\end{align}
Here ${}_{s}Y_{\ell m}$ denotes a spin-weighted spherical harmonic,
$u=t-x$ is retarded time, and the master variables are evaluated at
future null infinity.
}   In the present work, however, we work directly
with the master variables and focus on their resolvent, scattering, and
source--response structure.  The reconstruction of gravitational-wave strain and flux is left for future work.

\subsection{Frequency-domain outgoing resolvent}
\label{sec:outgoing_resolvent}

Equation~\eqref{eq:rwz_time_domain} separates into independent
$(\ell,m,\parity)$ channels.  We therefore suppress $m$ and $\parity$ when no
confusion can arise and define
\begin{equation}
    H_{\ell}
    =
    -\frac{\rmd^2}{\rmd x^2}
    +V_{\ell}(x).
    \label{eq:radial_hamiltonian}
\end{equation}
Our first source model is an instantaneous impulse,
\begin{equation}
    S_{\ell}(t,x)
    =
    \delta(t)f_{\ell}(x),
    \label{eq:impulsive_source}
\end{equation}
where $f_{\ell}$ is a localized spatial profile.  With the Fourier
convention
\begin{align}
    \widetilde{\psi}_{\ell}(\omega,x)
    &=
    \int_{-\infty}^{\infty}
    \rmd t\,
    \rme^{\rmi\omega t}
    \psi_{\ell}(t,x),
    \\
    \psi_{\ell}(t,x)
    &=
    \frac{1}{2\pi}
    \int_{\mathcal C}
    \rmd\omega\,
    \rme^{-\rmi\omega t}
    \widetilde{\psi}_{\ell}(\omega,x),
    \label{eq:fourier_pair}
\end{align}
where $\mathcal C$ is the retarded inversion contour, the radial
equation becomes
\begin{equation}
    \left(\omega^2-H_{\ell}\right)
    \widetilde{\psi}_{\ell}(\omega)
    =
    f_{\ell}.
    \label{eq:frequency_domain_driven}
\end{equation}
For $\im\omega>0$, the inverse is uniquely defined.  Its boundary value
on the real axis is the retarded, or outgoing, resolvent\footnote{%
This setup should be distinguished from the usual Lippmann--Schwinger
scattering problem, in which an incident solution of the free Hamiltonian
is specified independently and the total state contains both the incident
homogeneous wave and the scattered wave. In the present work, we instead
consider a retarded source-driven problem with no independently prescribed
incoming radiation. Once the source and the retarded radiation conditions
are fixed, the solution is uniquely determined by the full resolvent, and
no arbitrary homogeneous contribution should be added. In particular,
our impulsive source
$S_{\ell}(t,r_*)=\delta(t)f_{\ell}(r_*)$
may equivalently be regarded as specifying an instantaneous perturbation
of the initial data. This prescription does not exclude quasinormal modes:
their contributions arise from the poles of the analytically continued
resolvent and are excited by the projection of $f_{\ell}$ onto the
corresponding resonant subspace. Thus, a source-free ringdown can already
be described within linear perturbation theory through nontrivial initial
data. At second and higher orders, the same resolvent structure can be
used with effective sources generated by nonlinear combinations of
lower-order perturbations.
}
\begin{equation}
    R_{\ell}(\omega)
    =
    \lim_{\epsilon\downarrow0}
    \left[
        (\omega+\rmi\epsilon)^2-H_{\ell}
    \right]^{-1},
    \qquad
    \widetilde{\psi}_{\ell}(\omega)
    =
    R_{\ell}(\omega)f_{\ell}.
    \label{eq:retarded_resolvent}
\end{equation}
Equivalently, one may use the spectral variable
$z=(\omega+\rmi0)^2$ and write $R_{\ell}(z)=(z-H_{\ell})^{-1}$.
No independent homogeneous contribution is required once retarded
analyticity and the two radiation conditions have been imposed.

The integral kernel
\begin{equation}
    G_{\ell}(\omega;x,x')
    =
    \braket{x|R_{\ell}(\omega)|x'}
    \label{eq:green_kernel_definition}
\end{equation}
can be expressed in terms of two Jost solutions.  We choose
$\vsup{u}{H}_{\ell}$ to be purely ingoing at the future horizon and
$u_{\ell}^{\infty}$ to be purely outgoing at spatial infinity,
\begin{align}
    \vsup{u}{H}_{\ell}(\omega,x)
    &\sim
    \rme^{-\rmi\omega x},
    &x&\rightarrow-\infty,
    \\
    u_{\ell}^{\infty}(\omega,x)
    &\sim
    \rme^{+\rmi\omega x},
    &x&\rightarrow+\infty.
\end{align}
Then
\begin{equation}
    G_{\ell}(\omega;x,x')
    =
    \frac{
        \vsup{u}{H}_{\ell}(\omega,x_<)
        u_{\ell}^{\infty}(\omega,x_>)
    }{
        W_{\ell}(\omega)
    },
    \label{eq:jost_green_function}
\end{equation}
where $x_<:=\min(x,x')$, $x_>:=\max(x,x')$, and
\begin{equation}
    W_{\ell}(\omega)
    =
    \vsup{u}{H}_{\ell}
    \frac{\rmd u_{\ell}^{\infty}}{\rmd x}
    -
    \frac{\rmd \vsup{u}{H}_{\ell}}{\rmd x}
    u_{\ell}^{\infty}
\end{equation}
is independent of $x$.  At spatial infinity,
\begin{equation}
    \vsup{u}{H}_{\ell}
    \sim
    \vsup{A}{in}_{\ell}
    \rme^{-\rmi\omega x}
    +
    \vsup{A}{out}_{\ell}
    \rme^{+\rmi\omega x}.
    \label{eq:jost_asymptotic_coefficients}
\end{equation}
A unit-amplitude wave incident from infinity therefore has
\begin{equation}
    \mathcal T_{\ell}(\omega)
    =
    \frac{1}{\vsup{A}{in}_{\ell}(\omega)},
    \qquad
    \mathcal R_{\ell}(\omega)
    =
    \frac{\vsup{A}{out}_{\ell}(\omega)}
    {\vsup{A}{in}_{\ell}(\omega)}.
    \label{eq:scattering_coefficients_jost}
\end{equation}
The QNM frequencies $\omega_n$ satisfy
\begin{equation}
    \vsup{A}{in}_{\ell}(\omega_n)
    =
    0,
    \qquad
    W_{\ell}(\omega_n)
    =
    0.
\end{equation}
Thus the QNM frequencies,
real-frequency scattering amplitudes, and source response are all
encoded in the same Green function.

\section{Complex scaling}
\label{sec:complex_scaling}

\subsection{Global and exterior contour deformations}
\label{sec:complex_contour}

The purpose of complex scaling in the present work is to represent the
outgoing resolvent, rather than only to locate its poles.  We write $x=r_*$ and
replace the real radial axis by an analytic contour $\zeta_{\theta}(x)$ in the
complex plane.  The global dilation used in our previous QNM and
continuum-level-density calculations is
\begin{equation}
    \zeta_{\theta}(x)
    =
    x\rme^{\rmi\theta},
    \qquad
    0<\theta<\frac{\pi}{2}.
    \label{eq:global_contour}
\end{equation}
For a smooth dilation-analytic source, including the Gaussian profile used
below within its angular domain of analyticity, this transformation gives the
most direct numerical starting point.

For sources or observables that should remain on the physical real axis, one
may instead use exterior complex scaling (ECS).  A piecewise-linear ECS contour
is
\begin{equation}
    \zeta_{\theta}(x)
    =
    \begin{cases}
        -x_0+\left(x+x_0\right)\rme^{\rmi\theta},
        &x<-x_0,
        \\
        x,
        &|x|\leq x_0,
        \\
        x_0+\left(x-x_0\right)\rme^{\rmi\theta},
        &x>x_0,
    \end{cases}
    \label{eq:ecs_contour}
\end{equation}
where the undeformed interval is chosen to contain the principal potential
barrier, the source support, and any matching surfaces.  In a numerical
implementation the two corners may be smoothed.  The finite-basis formulas
below apply to either global scaling or ECS.

Let $U_{\theta}$ denote the non-unitary contour map, including the Jacobian
factor appropriate to the chosen representation.  The deformed operator is
\begin{equation}
    H_{\ell,\theta}
    =
    U_{\theta}H_{\ell}U_{\theta}^{-1}.
    \label{eq:complex_scaled_hamiltonian}
\end{equation}
For the global contour in Eq.~\eqref{eq:global_contour},
\begin{equation}
    H_{\ell,\theta}
    =
    -\rme^{-2\rmi\theta}
    \frac{\rmd^2}{\rmd x^2}
    +V_{\ell}\!\left(\zeta_{\theta}(x)\right),
    \label{eq:global_csm_operator}
\end{equation}
where the analytic continuation of $r=r(x)$, the inverse tortoise-coordinate relation, is taken on the same branch as in
the QNM calculation.  The contour must remain inside a domain where the
continued potential and the transformed source are single-valued and where no
singularity is crossed.

The outgoing radiation condition becomes square integrability along the
deformed contour.  At spatial infinity,
\begin{equation}
    \rme^{+\rmi\omega x}
    \longmapsto
    \rme^{+\rmi\omega\zeta_{\theta}(x)},
\end{equation}
which decays along a globally rotated ray when
\begin{equation}
    \im\!\left(\omega\rme^{\rmi\theta}\right)>0.
    \label{eq:csm_decay_condition}
\end{equation}
The ingoing horizon factor decays on the left ray under the corresponding rotation. QNM profiles that grow exponentially on the real axis can therefore be represented by $L^2$ eigenvectors of $H_{\ell,\theta}$. In the $z=\omega^2$ plane the continuum is rotated by approximately $2\theta$,
whereas exposed resonance eigenvalues are independent of $\theta$ in the exact theory. See Appendix~\ref{sec:abc} for the spectral mechanism underlying the Aguilar--Balslev--Combes complex-scaling theorem.  Residual $\theta$ dependence at finite basis size is used as a convergence diagnostic.

\subsection{Finite-basis resolvent}
\label{sec:finite_basis_resolvent}

We discretize $H_{\ell,\theta}$ in a nonorthogonal set of localized
basis functions.  A convenient choice is a polynomial--Gaussian family
\begin{equation}
    \phi_{i n}(x)
    =
    \mathcal N_{i n}
    \left(\sqrt{\alpha_i}\,x\right)^n
    \exp\left(-\alpha_i x^2\right),
    \qquad
    \text{$n=0$, $1$, \dots, $p_{\max}$},
    \label{eq:polynomial_gaussian_basis}
\end{equation}
where $\mathcal N_{i n}$ denotes a normalization constant and
$\alpha_i=r_i^{-2}$.  The length scales are distributed geometrically,
\begin{equation}
    r_i
    =
    r_{\min}a^{i-1},
    \qquad
    a
    =
    \left(\frac{r_{\max}}{r_{\min}}\right)^{1/(N_r-1)}.
    \label{eq:gaussian_ranges}
\end{equation}
The polynomial degree controls local shape, whereas the geometric set
of ranges resolves both the potential barrier and the extended tails.
This basis was used previously for pole and continuum-level-density
calculations, but here the primary numerical object is the resolvent
matrix rather than the eigenvalue list alone.

With a collective basis index $\gamma=(i,n)$, define
\begin{equation}
    \left(\mathbf H_{\theta}\right)_{\gamma\gamma'}
    =
    \braket{\phi_\gamma|H_{\ell,\theta}|\phi_{\gamma'}},
    \qquad
    \left(\mathbf N\right)_{\gamma\gamma'}
    =
    \braket{\phi_\gamma|\phi_{\gamma'}}.
\end{equation}
The right and left generalized eigenvectors satisfy
\begin{align}
    \mathbf H_{\theta}\mathbf c_{\nu}^{R}
    &=
    E_{\nu}^{\theta}
    \mathbf N\mathbf c_{\nu}^{R},
    \label{eq:generalized_right_eigenproblem}
    \\
    \left(\mathbf c_{\nu}^{L}\right)^{\dagger}
    \mathbf H_{\theta}
    &=
    E_{\nu}^{\theta}
    \left(\mathbf c_{\nu}^{L}\right)^{\dagger}
    \mathbf N,
    \label{eq:generalized_left_eigenproblem}
\end{align}
and are normalized according to
\begin{equation}
    \left(\mathbf c_{\mu}^{L}\right)^{\dagger}
    \mathbf N
    \mathbf c_{\nu}^{R}
    =
    \delta_{\mu\nu}.
    \label{eq:biorthogonal_normalization}
\end{equation}
For a complex-symmetric matrix representation, this relation may be
implemented with the associated $c$ product.  We retain left and right
notation because it also applies to more general discretizations.

Within the finite basis, the driven equation is solved directly by
\begin{equation}
    \mathbf c(z)
    =
    \left(z\mathbf N-\mathbf H_{\theta}\right)^{-1}
    \mathbf b,
    \qquad
    b_\gamma
    =
    \braket{\phi_\gamma|U_{\theta}|f_{\ell}}.
    \label{eq:matrix_driven_solution}
\end{equation}
This inverse maps a covector of Galerkin right-hand sides to the coefficient
vector.  The matrix representing the resolvent as an operator on coefficient
vectors is instead
\begin{equation}
    \mathbf G_{\theta}(z)
    :=
    \left(z\mathbf N-\mathbf H_{\theta}\right)^{-1}
    \mathbf N.
    \label{eq:coefficient_space_resolvent}
\end{equation}
With the normalization in Eq.~\eqref{eq:biorthogonal_normalization}, the two
matrices admit the expansions
\begin{equation}
    \left(z\mathbf N-\mathbf H_{\theta}\right)^{-1}
    =
    \sum_{\nu=1}^{N}
    \frac{
        \mathbf c_{\nu}^{R}
        \left(\mathbf c_{\nu}^{L}\right)^{\dagger}
    }{
        z-E_{\nu}^{\theta}
    },
    \label{eq:finite_spectral_resolvent}
\end{equation}
and
\begin{equation}
    \mathbf G_{\theta}(z)
    =
    \sum_{\nu=1}^{N}
    \frac{
        \mathbf c_{\nu}^{R}
        \left(\mathbf c_{\nu}^{L}\right)^{\dagger}
        \mathbf N
    }{
        z-E_{\nu}^{\theta}
    },
    \label{eq:finite_operator_resolvent}
\end{equation}
where $N=N_r\times(p_{\max}+1)$.
In the continuum limit, this sum approaches the extended completeness
relation containing bound states, exposed resonances, and an integral
along the rotated continuum.  At finite $N$, the latter is represented
by complex pseudostates aligned with the deformed cut
\cite{Myo:2014ypa}.

For $z$ in the initial resolvent domain, and thereafter by analytic
continuation between dilation-analytic states, the physical and
deformed resolvents are related by
\begin{equation}
    R_{\ell}(z)
    =
    U_{\theta}^{-1}
    R_{\ell,\theta}(z)
    U_{\theta},
    \qquad
    R_{\ell,\theta}(z)
    =
    \left(z-H_{\ell,\theta}\right)^{-1}.
    \label{eq:resolvent_similarity_relation}
\end{equation}

This relation clarifies the status of complex scaling in the present
construction.  It is not a physical deformation of the black-hole potential
and it does not replace the outgoing problem by a different bound-state
problem.  Rather, it is a representation of the same analytically continued
resolvent in which outgoing resonant states become square integrable and the
continuum is rotated away from them.  Consequently, $\theta$ is not an
observable parameter.  Any residual dependence on $\theta$ after basis
truncation measures numerical error or a failure to remain in a common
analyticity domain.  Likewise, the pseudostates that discretize the rotated
continuum are quadrature data for the non-pole response, not additional
physical QNMs.

For arbitrary source and observation states $f$ and $g$, the resulting
matrix element is
\begin{align}
    \braket{g|R_{\ell}(z)|f}
    &=
    \sum_{\nu}
    \frac{
        \braket{g|U_{\theta}^{-1}|\psi_{\nu,\theta}^{R}}
        \braket{\psi_{\nu,\theta}^{L}|U_{\theta}|f}
    }{
        z-E_{\nu}^{\theta}
    }.
    \label{eq:source_observer_factorization}
\end{align}
For global scaling, both the source and observation states in
Eq.~\eqref{eq:source_observer_factorization} must be transformed, and the
formula is restricted to dilation-analytic vectors.  The Gaussian source used
in Sec.~\ref{sec:source_response} satisfies this requirement for an admissible
range of $\theta$.

Equation~\eqref{eq:source_observer_factorization} should not be confused
with a term-by-term time evolution in the variable $E_{\nu}^{\theta}$.
The wave equation is second order in time, and the map
$z=\omega^2$ introduces two square-root sheets.  The retarded time-domain
field must therefore be reconstructed from the $\omega$-plane integral
in Eq.~\eqref{eq:fourier_pair}.  In particular, a simple QNM pole at
$\omega=\omega_n$ contributes
\begin{equation}
    \psi_{\ell,n}(t,x)
    =
    -\rmi\,\Theta(t)
    \rme^{-\rmi\omega_n t}
    \Res_{\omega=\omega_n}
    \widetilde{\psi}_{\ell}(\omega,x),
    \label{eq:qnm_residue_time_domain}
\end{equation}
with the symmetry-related pole at $-\omega_n^*$ included when a real field is reconstructed. Here $\Theta(t)$ is the Heaviside step function. The remaining contour contribution contains the prompt and continuum parts of the signal.

\section{Rank-two Riesz clusters and real-frequency scattering}
\label{sec:real_axis_ep}

Reconstructing a known Schwarzschild greybody curve from the Green function
would principally validate the channel-resolved resolvent rather than provide
an independent physical result.  Our purpose in this section is instead
analytic.  We first derive a representation of an isolated QNM pair that is
regular through a second-order exceptional point (EP), and then apply it to
the transmission amplitude.  This separates two statements that are sometimes
conflated: the modal decomposition is singular at an EP, whereas a physical
real-axis matrix element need not be.  The distinction is particularly
relevant because QNM spectra can be highly sensitive to a potential
deformation even when the associated greybody factors remain comparatively
stable~\cite{Xie:2025jbr}.
Related distinctions between QNM spectral instability and physical response
have been studied in Refs.~\cite{Jaramillo:2020tuu,Yang:2024vor}.

Two analyticity questions enter this section.  At fixed operator parameters
$q$, the continued resolvent is singular as a function of $z$ at a QNM pole,
and the EP produces a double pole at $z=\varEP{E}$.  By contrast, in a
real-frequency scattering experiment one fixes $z=z_{\omega}$ and varies the
operator parameters through $q=\varEP{q}$.  If $z_{\omega}$ remains away from
the continued pole, the resulting channel matrix element can be regular in
$q$ even though the individual QNM branches are not.  The regularity result
below concerns this second question and does not remove the first singularity.

\subsection{On-shell channel matrix elements}
\label{sec:on_shell_channels}

Let $q$ denote real parameters of a real radial potential; for the Gaussian
family introduced in Sec.~\ref{sec:gaussian_bump_ep},
$q=(\varepsilon,d)$ at fixed width $\sigma$.  For real $\omega>0$,
flux-normalized scattering from spatial infinity is defined by
\begin{align}
    \vsup{\psi}{sc}_{\ell\omega}(x;q)
    &\sim
    \rme^{-\rmi\omega x}
    +\mathcal R_{\ell}(\omega;q)
    \rme^{+\rmi\omega x},
    &x&\rightarrow+\infty,
    \\
    \vsup{\psi}{sc}_{\ell\omega}(x;q)
    &\sim
    \mathcal T_{\ell}(\omega;q)
    \rme^{-\rmi\omega x},
    &x&\rightarrow-\infty.
\end{align}
The greybody factor and the flux identity are
\begin{align}
    \Gamma_{\ell}(\omega;q)
    &=
    \left|\mathcal T_{\ell}(\omega;q)\right|^2,
    \label{eq:greybody_factor_parameterized}
    \\
    \left|\mathcal R_{\ell}(\omega;q)\right|^2
    +
    \left|\mathcal T_{\ell}(\omega;q)\right|^2
    &=
    1.
    \label{eq:greybody_unitarity}
\end{align}
The second relation holds for the real, nonrotating potentials considered
here and supplies a stringent numerical check.

To expose the operator structure, define the linear transmission functional
\begin{align}
    \mathfrak T_{\ell,\omega}[A]
    &:=
    2\rmi\omega
    \lim_{\substack{x\rightarrow-\infty\\x'\rightarrow+\infty}}
    \rme^{+\rmi\omega x}
    \rme^{-\rmi\omega x'}
    \braket{x|A|x'}.
    \label{eq:transmission_functional}
\end{align}
For the retarded boundary value
$z_{\omega}:=(\omega+\rmi0)^2$, the Jost representation in
Eq.~\eqref{eq:jost_green_function} gives
\begin{equation}
    \mathcal T_{\ell}(\omega;q)
    =
    \mathfrak T_{\ell,\omega}
    \left[R_{\ell}(z_{\omega};q)\right].
    \label{eq:transmission_from_green}
\end{equation}
For completeness, with
\begin{equation}
    G_{0}(\omega;x,x')
    =
    \frac{\rme^{\rmi\omega|x-x'|}}{2\rmi\omega},
    \label{eq:free_green_function}
\end{equation}
the reflection amplitude is
\begin{align}
    \mathcal R_{\ell}(\omega;q)
    &=
    2\rmi\omega
    \lim_{x,x'\rightarrow+\infty}
    \rme^{-\rmi\omega(x+x')}
    \left[
        G_{\ell}(\omega;x,x';q)
        -G_{0}(\omega;x,x')
    \right].
    \label{eq:reflection_from_green}
\end{align}
In an ECS calculation these channel functionals are implemented by matching
the Green function and its radial derivative at two surfaces in the
undeformed region.  The surfaces must lie outside the principal interaction
region, and their displacement provides an estimate of the residual
asymptotic error.

The need for a channel functional is important.  The continuum level density
contains only trace information,
\begin{equation}
    \Delta\rho(E)
    =
    -\frac{1}{\pi}
    \im\Tr
    \left[
        R(E+\rmi0)-R_0(E+\rmi0)
    \right],
    \label{eq:cld_definition_recap}
\end{equation}
and cannot determine the modulus of an off-diagonal transmission amplitude.
The role of Eqs.~\eqref{eq:transmission_from_green} and
\eqref{eq:reflection_from_green} is therefore to turn the same
complex-scaled resolvent into a physical channel observable.

\subsection{Two contour moments and the exact pair resolvent}
\label{sec:rank_two_riesz_cluster}

Assume that two QNM eigenvalues form an isolated algebraic cluster of total
multiplicity two and coalesce at $q=\varEP{q}$.  Let
$\vsub{\Gamma}{pair}$ be a positively oriented contour in the $z$ plane that
encloses this cluster and no other spectral point of the complex-scaled
operator.  The first two contour moments are
\begin{align}
    \vsub{P}{pair}^{\theta}(q)
    &:=
    \frac{1}{2\pi\rmi}
    \oint_{\vsub{\Gamma}{pair}}
    R_{\ell,\theta}(z;q)\,\rmd z,
    \label{eq:ep_riesz_projector}
    \\
    \vsub{\mathcal M}{pair}^{\theta}(q)
    &:=
    \frac{1}{2\pi\rmi}
    \oint_{\vsub{\Gamma}{pair}}
    zR_{\ell,\theta}(z;q)\,\rmd z
    =
    H_{\ell,\theta}(q)
    \vsub{P}{pair}^{\theta}(q).
    \label{eq:cluster_first_moment}
\end{align}
The first is the Riesz projector onto the root subspace; the second is the
restriction of the spectral operator to the same subspace, embedded in the
full space.  Because this subspace has dimension two, define its spectral
center and traceless part by
\begin{align}
    \vsub{E}{c}(q)
    &:=
    \frac{1}{2}
    \Tr\vsub{\mathcal M}{pair}^{\theta}(q),
    \label{eq:cluster_spectral_center}
    \\
    \vsub{K}{pair}^{\theta}(q)
    &:=
    \vsub{\mathcal M}{pair}^{\theta}(q)
    -\vsub{E}{c}(q)
    \vsub{P}{pair}^{\theta}(q).
    \label{eq:cluster_traceless_part}
\end{align}
The traces are finite-dimensional traces on the range of
$\vsub{P}{pair}^{\theta}$, equivalently the ordinary operator traces of these
finite-rank operators.  Finally, introduce the single-valued squared
half-splitting
\begin{equation}
    \Delta(q)^2
    :=
    \frac{1}{2}
    \Tr
    \left[
        \vsub{K}{pair}^{\theta}(q)^2
    \right].
    \label{eq:cluster_discriminant}
\end{equation}
Only $\Delta^2$, rather than a chosen square-root branch $\Delta$, is needed
below.

The information content of the two moments has a simple finite-dimensional
interpretation.  On the range of $\vsub{P}{pair}^{\theta}$, the projector acts
as the two-dimensional identity, while
$\vsub{\mathcal M}{pair}^{\theta}=H_{\ell,\theta}
\vsub{P}{pair}^{\theta}$ is the matrix of the restricted operator.  The
projector therefore identifies the invariant root space, and the first moment
specifies how the operator acts inside that space.  Taking the trace and
traceless part of this restricted matrix gives $\vsub{E}{c}$ and
$\vsub{K}{pair}^{\theta}$.  This is why the projector alone is insufficient
at an EP: it retains the two-dimensional space but not the nilpotent action
that produces the double pole.

\begin{proposition}[Exact rank-two cluster resolvent]
\label{prop:exact_rank_two_cluster}
The resolvent restricted to the isolated pair is
\begin{align}
    \vsub{R}{pair}^{\theta}(z;q)
    &:=
    R_{\ell,\theta}(z;q)
    \vsub{P}{pair}^{\theta}(q)
    \label{eq:ep_pair_resolvent}
    \\
    &=
    \frac{
        \left[z-\vsub{E}{c}(q)\right]
        \vsub{P}{pair}^{\theta}(q)
        +\vsub{K}{pair}^{\theta}(q)
    }{
        \left[z-\vsub{E}{c}(q)\right]^2
        -\Delta(q)^2
    }.
    \label{eq:exact_pair_resolvent}
\end{align}
It is determined completely by the two contour moments in
Eqs.~\eqref{eq:ep_riesz_projector} and
\eqref{eq:cluster_first_moment}.
\end{proposition}

\begin{proof}
For compactness, suppress the parameter argument $q$ and write
\begin{align}
    H
    &:=
    H_{\ell,\theta}(q),
    \\
    P
    &:=
    \vsub{P}{pair}^{\theta}(q),
    \\
    \mathcal M
    &:=
    \vsub{\mathcal M}{pair}^{\theta}(q),
    \\
    E_c
    &:=
    \vsub{E}{c}(q),
    \\
    K
    &:=
    \vsub{K}{pair}^{\theta}(q).
\end{align}
By the holomorphic functional calculus, the Riesz projector satisfies
\begin{equation}
    P^2=P,
    \qquad
    HP=PH.
\end{equation}
Together with $\mathcal M=HP$ and $K=\mathcal M-E_cP$, this also implies
\begin{equation}
    P\mathcal M
    =
    \mathcal MP
    =
    \mathcal M,
    \qquad
    PK
    =
    KP
    =
    K.
\end{equation}
It follows that the two-dimensional subspace
\begin{equation}
    \mathcal X
    :=
    \operatorname{im}(P)
\end{equation}
is invariant under $H$.\footnote{%
We denote the image of a map $f$ by $\operatorname{im}(f)$.} On $\mathcal X$, the projector $P$ is the identity
$I_{\mathcal X}$, and $\mathcal M=HP$ restricts to the operator
$H|_{\mathcal X}$.  The definition of $K$ may therefore be written as
\begin{equation}
    K|_{\mathcal X}
    =
    H|_{\mathcal X}-E_c I_{\mathcal X}.
\end{equation}
Moreover, since $E_c=\frac{1}{2}\Tr\mathcal M$ and
$\dim\mathcal X=2$, its trace on $\mathcal X$ vanishes:
\begin{align}
    \Tr_{\mathcal X}
    \left(K|_{\mathcal X}\right)
    &=
    \Tr_{\mathcal X}
    \left(H|_{\mathcal X}\right)
    -E_c\Tr_{\mathcal X}I_{\mathcal X}
    \\
    &=
    2E_c-2E_c
    =0.
\end{align}

For any operator $B$ on a two-dimensional space, the
Cayley--Hamilton identity is
\begin{equation}
    B^2
    -\left(\Tr B\right)B
    +\left(\det B\right)I
    =0.
\end{equation}
Applying it to $B=K|_{\mathcal X}$ and using the vanishing trace gives
\begin{equation}
    \left(K|_{\mathcal X}\right)^2
    =
    -\det\left(K|_{\mathcal X}\right)I_{\mathcal X}.
\end{equation}
For a two-dimensional operator, its determinant can also be expressed as
\begin{equation}
    \det B
    =
    \frac{1}{2}
    \left[
        \left(\Tr B\right)^2
        -\Tr\left(B^2\right)
    \right].
\end{equation}
Consequently,
\begin{align}
    \det\left(K|_{\mathcal X}\right)
    &=
    -\frac{1}{2}
    \Tr_{\mathcal X}
    \left[
        \left(K|_{\mathcal X}\right)^2
    \right]
    \\
    &=
    -\Delta^2.
\end{align}
Here the trace on $\mathcal X$ agrees with the full operator trace because
$K$ is supported on the finite-dimensional range of $P$.  Thus, as an
identity embedded in the full space,
\begin{equation}
    K^2
    =
    \Delta^2P.
    \label{eq:cluster_cayley_hamilton}
\end{equation}

Now define
\begin{equation}
    a(z)
    :=
    z-E_c,
    \qquad
    D(z)
    :=
    a(z)^2-\Delta^2.
\end{equation}
From $HP=E_cP+K$, one obtains
\begin{equation}
    (z-H)P
    =
    a(z)P-K.
\end{equation}
Since $PK=KP=K$, Eq.~\eqref{eq:cluster_cayley_hamilton} gives
\begin{align}
    &\left[a(z)P-K\right]
    \left[a(z)P+K\right]
    \\
    &\qquad=
    a(z)^2P-K^2
    \\
    &\qquad=
    D(z)P.
    \label{eq:cluster_resolvent_product}
\end{align}
The same identity holds with the two factors interchanged.  Therefore, when
$D(z)\neq0$, the operator
\begin{equation}
    Q(z)
    :=
    \frac{a(z)P+K}{D(z)}
\end{equation}
satisfies
\begin{equation}
    (z-H)Q(z)
    =
    Q(z)(z-H)
    =
    P.
\end{equation}
It is thus the inverse of $z-H$ on $\mathcal X$, embedded in the full space.
Because $P$ commutes with the resolvent, this inverse is precisely
\begin{equation}
    Q(z)
    =
    (z-H)^{-1}P
    =
    R_{\ell,\theta}(z;q)P.
\end{equation}
Restoring the full notation gives Eq.~\eqref{eq:exact_pair_resolvent}.  The
identity holds away from the cluster poles and, as a meromorphic identity,
determines the behavior at those poles as well.  Finally, $P$ and $\mathcal M$
are exactly the zeroth and first contour moments, while $E_c$, $K$, and
$\Delta^2$ are obtained from them by Eqs.~\eqref{eq:cluster_spectral_center}--
\eqref{eq:cluster_discriminant}.  Hence no spectral data beyond the two contour
moments are required.
\end{proof}

The formula is exact for the spectral component selected by
$\vsub{\Gamma}{pair}$; it is not a phenomenological two-mode truncation of
the full response.  The only decomposition is the identity
$R_{\ell,\theta}=\vsub{R}{pair}^{\theta}+
\vsub{R}{rest}^{\theta}$, with the remainder retained explicitly.  Away from
the EP, the quadratic denominator factors into the two simple-pole
denominators.  At the EP, $\Delta^2$ vanishes while
$\vsub{K}{pair}^{\theta}$ remains nonzero and nilpotent, so the same formula
becomes the simple-plus-double-pole Laurent form.  No change of ansatz is
required at the coalescence.

Away from the EP, choose either local branch of $\Delta$.  The two eigenvalues
and their projectors are
\begin{align}
    E_{\pm}(q)
    &=
    \vsub{E}{c}(q)\pm\Delta(q),
    \\
    P_{\pm}^{\theta}(q)
    &=
    \frac{1}{2}
    \left[
        \vsub{P}{pair}^{\theta}(q)
        \pm
        \frac{\vsub{K}{pair}^{\theta}(q)}{\Delta(q)}
    \right].
    \label{eq:individual_projectors_from_cluster}
\end{align}
A detailed derivation and direct verification of these projectors are given in Appendix~\ref{app:individual_projectors}.
Equation~\eqref{eq:individual_projectors_from_cluster} displays explicitly why
the individual projectors can contain opposite $1/\Delta$ terms, while their
sum and the exact pair resolvent remain finite.  The branch exchange resides in
$\Delta$; the quantities $\vsub{P}{pair}^{\theta}$,
$\vsub{K}{pair}^{\theta}$, $\vsub{E}{c}$, and $\Delta^2$ are single-valued for
an analytic operator family as long as the contour continues to isolate the
same cluster.

This is the local operator counterpart of the path-dependent spectral
reorganization found in Kerr EP cascades
\cite{Cavalcante:2025abr}:
a permutation of the two branches enclosed by a fixed isolating contour
does not change
the pair projector or the operator restricted to that root subspace.

\paragraph{Relation to the mean QNM frequency}
The spectral center $\vsub{E}{c}$ should be distinguished from the mean
QNM frequency often used to characterize a resonant pair in black-hole
perturbation theory
\cite{Motohashi:2024fwt,Yang:2025dbn,PanossoMacedo:2025xnf}.
Away from the EP, let
\begin{equation}
    E_{\pm}
    =
    \omega_{\pm}^{2},
    \qquad
    \overline{\omega}
    :=
    \frac{\omega_{+}+\omega_{-}}{2},
    \qquad
    \delta\omega
    :=
    \frac{\omega_{+}-\omega_{-}}{2}.
\end{equation}
Then the two center variables are related by
\begin{align}
    \vsub{E}{c}
    &=
    \frac{E_{+}+E_{-}}{2}
    \notag\\
    &=
    \overline{\omega}^{\,2}
    +
    \delta\omega^{2}.
    \label{eq:spectral_center_frequency_mean}
\end{align}
Thus $\vsub{E}{c}$ is not, away from the EP, the square of the mean QNM
frequency.  The quantity $\overline{\omega}$ is natural for describing the
oscillation frequency of a near-resonant time-domain signal, whereas
$\vsub{E}{c}$ is the canonical center of the present $z$-plane resolvent:
it is obtained directly from a Riesz contour moment and enters the reduced
characteristic polynomial without assigning labels to the two QNM branches.
At the EP, $\delta\omega=0$, and hence
\begin{equation}
    \vsub{E}{c}(\varEP{q})
    =
    \varEP{\omega}^{\,2}
\end{equation}
on the QNM sheet under consideration.  The two descriptions therefore agree
at the coalescence while organizing the neighborhood of the EP in different
spectral variables.

At a second-order EP,
\begin{equation}
    \Delta(\varEP{q})^2=0,
    \qquad
    \vsub{K}{pair}^{\theta}(\varEP{q})\neq0,
    \qquad
    \vsub{K}{pair}^{\theta}(\varEP{q})^2=0.
    \label{eq:cluster_ep_conditions}
\end{equation}
The exact formula therefore becomes
\begin{equation}
    \vsub{R}{pair}^{\theta}(z;\varEP{q})
    =
    \frac{\vsub{K}{pair}^{\theta}(\varEP{q})}
         {(z-\varEP{E})^2}
    +
    \frac{\vsub{P}{pair}^{\theta}(\varEP{q})}
         {z-\varEP{E}},
    \qquad
    \varEP{E}=\vsub{E}{c}(\varEP{q}).
    \label{eq:pair_resolvent_ep_from_moments}
\end{equation}
Thus the nilpotent operator multiplying the double pole is obtained from the
same two contour moments as the simple-pole coefficient, without constructing
a Jordan chain.

The complementary resolvent is
\begin{equation}
    \vsub{R}{rest}^{\theta}(z;q)
    :=
    R_{\ell,\theta}(z;q)
    -\vsub{R}{pair}^{\theta}(z;q).
    \label{eq:ep_rest_resolvent}
\end{equation}
Pulling the pair and rest sectors back to the physical contour gives
\begin{equation}
    R_{\varkappa}(z;q)
    =
    U_{\theta}^{-1}
    R_{\varkappa}^{\theta}(z;q)
    U_{\theta},
    \qquad
    \varkappa\in\{\mathrm{pair},\mathrm{rest}\}.
    \label{eq:physical_pair_rest_resolvents}
\end{equation}

\begin{proposition}[Scaling-angle independence of physical moments]
Let $f$ and $g$ be dilation-analytic vectors, and let the deformation angle
vary without crossing a singularity or changing the cluster enclosed by
$\vsub{\Gamma}{pair}$.  Then
\begin{align}
    p_{g,f}(q)
    &:=
    \braket{
        g|U_{\theta}^{-1}
        \vsub{P}{pair}^{\theta}(q)
        U_{\theta}|f
    }
    \\
    &=
    \frac{1}{2\pi\rmi}
    \oint_{\vsub{\Gamma}{pair}}
    \braket{g|R_{\ell}(z;q)|f}\,\rmd z,
    \label{eq:physical_projector_moment}
    \\
    m_{g,f}(q)
    &:=
    \braket{
        g|U_{\theta}^{-1}
        \vsub{\mathcal M}{pair}^{\theta}(q)
        U_{\theta}|f
    }
    \\
    &=
    \frac{1}{2\pi\rmi}
    \oint_{\vsub{\Gamma}{pair}}
    z\braket{g|R_{\ell}(z;q)|f}\,\rmd z.
    \label{eq:physical_first_moment}
\end{align}
These scalar moments are independent of $\theta$.
\end{proposition}

\begin{proof}
Insert the similarity relation~\eqref{eq:resolvent_similarity_relation}
inside the two contour integrals.  Their right-hand sides are contour moments
of the unique meromorphic continuation of the physical scalar resolvent and
contain no reference to the deformation angle.
\end{proof}

It is also useful to introduce the matrix element of the traceless cluster
operator,
\begin{align}
    k_{g,f}(q)
    &:=
    m_{g,f}(q)
    -
    \vsub{E}{c}(q)p_{g,f}(q)
    \notag\\
    &=
    \braket{
        g
        |
        U_{\theta}^{-1}
        \vsub{K}{pair}^{\theta}(q)
        U_{\theta}
        |
        f
    }.
    \label{eq:physical_traceless_cluster_moment}
\end{align}

For scattering channels, plane waves themselves are not dilation-analytic
$L^2$ vectors.  The transition-operator form below instead couples the
resolvent to the localized vectors
$\mathsf V_{\ell}\ket{-\omega}$; equivalently, ECS permits the channel
matching to be performed in an undeformed region.  These are the appropriate
real-frequency realizations of the same physical contour moments.

\subsection{Pair/rest transmission amplitude}
\label{sec:pair_resolved_transmission}

For the sector decomposition it is useful to pass from the Green kernel to
the transition operator.\footnote{%
The construction in
Eqs.~\eqref{eq:transition_operator}--%
\eqref{eq:transmission_from_transition_operator}
is the standard Lippmann--Schwinger formulation of one-dimensional
stationary scattering~\cite{Taylor:1972}, expressed using our resolvent convention
$R_{\ell}(z;q)=[z-H_{\ell}(q)]^{-1}$ and channel normalization.
We include it to fix conventions and to connect the complex-scaled
resolvent explicitly with the on-shell transmission amplitude.
The EP-specific step is the subsequent Riesz decomposition into
$\vsub{R}{pair}$ and $\vsub{R}{rest}$, which reorganizes the
standard scattering amplitude into contributions that remain
well-defined through the mode coalescence.
}  Let
\begin{equation}
    H_0
    =
    -\frac{\rmd^2}{\rmd x^2},
    \qquad
    \mathsf V_{\ell}(q)
    :=
    H_{\ell}(q)-H_0,
\end{equation}
and define
\begin{equation}
    \mathsf T_{\ell}(z;q)
    =
    \mathsf V_{\ell}(q)
    +
    \mathsf V_{\ell}(q)
    R_{\ell}(z;q)
    \mathsf V_{\ell}(q).
    \label{eq:transition_operator}
\end{equation}
The resolvent identity then reads
\begin{equation}
    R_{\ell}
    =
    R_0
    +R_0\mathsf T_{\ell}R_0.
    \label{eq:resolvent_transition_identity}
\end{equation}
For plane waves normalized by
$\braket{x|{-\omega}}=\rme^{-\rmi\omega x}$, the full transmission amplitude
can equivalently be written as
\begin{equation}
    \mathcal T_{\ell}(\omega;q)
    =
    1
    +
    \frac{1}{2\rmi\omega}
    \braket{
        {-\omega}
        |\mathsf T_{\ell}(z_{\omega};q)|
        {-\omega}
    }.
    \label{eq:transmission_from_transition_operator}
\end{equation}

Equations~\eqref{eq:ep_pair_resolvent} and
\eqref{eq:ep_rest_resolvent} induce
\begin{align}
    \mathsf T_{\ell,\mathrm{pair}}(z;q)
    &:=
    \mathsf V_{\ell}(q)
    R_{\mathrm{pair}}(z;q)
    \mathsf V_{\ell}(q),
    \\
    \mathsf T_{\ell,\mathrm{rest}}(z;q)
    &:=
    \mathsf V_{\ell}(q)
    +
    \mathsf V_{\ell}(q)
    R_{\mathrm{rest}}(z;q)
    \mathsf V_{\ell}(q).
    \label{eq:transition_pair_rest_sectors}
\end{align}
The free-transmission term is assigned to the rest sector.  We then obtain
\begin{align}
    \mathcal T_{\ell}(\omega;q)
    &=
    \mathcal T_{\ell,\mathrm{pair}}(\omega;q)
    +
    \mathcal T_{\ell,\mathrm{rest}}(\omega;q),
    \label{eq:transmission_pair_rest_split}
    \\
    \mathcal T_{\ell,\mathrm{pair}}(\omega;q)
    &:=
    \frac{1}{2\rmi\omega}
    \braket{
        {-\omega}
        |\mathsf T_{\ell,\mathrm{pair}}(z_{\omega};q)|
        {-\omega}
    },
    \\
    \mathcal T_{\ell,\mathrm{rest}}(\omega;q)
    &:=
    1
    +
    \frac{1}{2\rmi\omega}
    \braket{
        {-\omega}
        |\mathsf T_{\ell,\mathrm{rest}}(z_{\omega};q)|
        {-\omega}
    }.
    \label{eq:transmission_sector_definition}
\end{align}

The pair/rest split is a coherent decomposition of an amplitude, not a
division into two mutually exclusive scattering processes.  The pair sector
contains the contribution of the chosen Riesz cluster, whereas the rest
sector contains the free-transmission term, all other poles, and the
continuum contribution.  Deforming $\vsub{\Gamma}{pair}$ without crossing the
spectrum leaves the split unchanged, but only the sum is the complete
transmission amplitude.  The separate terms are useful diagnostics precisely
because their interference is retained.

This form couples the spectral resolvent to the localized channel vectors
$\mathsf V_{\ell}\ket{-\omega}$ and is therefore preferable to applying an
asymptotic limit separately to back-rotated QNM projectors.
The greybody factor is not a sum of two partial probabilities.  Rather,
\begin{align}
    \Gamma_{\ell}(\omega;q)
    &=
    \left|
        \mathcal T_{\ell,\mathrm{pair}}
    \right|^2
    +
    \left|
        \mathcal T_{\ell,\mathrm{rest}}
    \right|^2
    \notag\\
    &\quad
    +2\re
    \left[
        \mathcal T_{\ell,\mathrm{pair}}
        \left(
            \mathcal T_{\ell,\mathrm{rest}}
        \right)^{*}
    \right],
    \label{eq:greybody_pair_rest_interference}
\end{align}
where the common arguments $(\omega;q)$ are suppressed on the right-hand
side.  The interference term is essential: large individual resonant
contributions need not imply a large change in the full transmission.

Applying the exact cluster formula~\eqref{eq:exact_pair_resolvent} to the two
localized channel vectors gives
\begin{equation}
    \mathcal T_{\ell,\mathrm{pair}}(\omega;q)
    =
    \frac{
        \left[z_{\omega}-\vsub{E}{c}(q)\right]
        \vsub{\tau}{P}(\omega;q)
        +\vsub{\tau}{K}(\omega;q)
    }{
        \left[z_{\omega}-\vsub{E}{c}(q)\right]^2
        -\Delta(q)^2
    }.
    \label{eq:exact_pair_transmission}
\end{equation}
Here the two channel moments are
\begin{align}
    \vsub{\tau}{P}(\omega;q)
    &:=
    \frac{1}{2\rmi\omega}
    \braket{
        {-\omega}
        |\mathsf V_{\ell}(q)
        U_{\theta}^{-1}
        \vsub{P}{pair}^{\theta}(q)
        U_{\theta}
        \mathsf V_{\ell}(q)|
        {-\omega}
    },
    \\
    \vsub{\tau}{K}(\omega;q)
    &:=
    \frac{1}{2\rmi\omega}
    \braket{
        {-\omega}
        |\mathsf V_{\ell}(q)
        U_{\theta}^{-1}
        \vsub{K}{pair}^{\theta}(q)
        U_{\theta}
        \mathsf V_{\ell}(q)|
        {-\omega}
    }.
    \label{eq:pair_channel_moments}
\end{align}
This expression is valid both away from and at the EP.  Away from the EP it
equals the sum of the two simple-pole channel amplitudes, but it does not
require either individual projector in
Eq.~\eqref{eq:individual_projectors_from_cluster}.

At the EP, Eqs.~\eqref{eq:pair_resolvent_ep_from_moments} and
\eqref{eq:exact_pair_transmission} give
\begin{equation}
    \mathcal T_{\ell,\mathrm{pair}}
    (\omega;\varEP{q})
    =
    \frac{\vsub{\tau}{K}(\omega;\varEP{q})}
    {(z_{\omega}-\varEP{E})^2}
    +
    \frac{\vsub{\tau}{P}(\omega;\varEP{q})}
    {z_{\omega}-\varEP{E}}.
    \label{eq:ep_pair_transmission_laurent}
\end{equation}
Because the channel vectors and the prefactor in
Eq.~\eqref{eq:pair_channel_moments} depend on $\omega$, the two numerators in
Eq.~\eqref{eq:ep_pair_transmission_laurent} should not themselves be called
Laurent coefficients in either $z$ or $\omega$.  The operator-valued Laurent
coefficients are $\vsub{P}{pair}^{\theta}$ and
$\vsub{K}{pair}^{\theta}(\varEP{q})$; scalar Laurent coefficients are obtained
only after expanding the complete on-shell matrix element about
$\omega=\varEP{\omega}$.
Since the QNM value $\varEP{E}$ lies off the positive real axis, neither
denominator vanishes for real $\omega$.  Thus the double pole governs the
analytic continuation while the on-shell pair amplitude remains finite.

\subsection{Regularity of physical scattering across the EP}
\label{sec:ep_scattering_diagnostic}

To compare the complex spectrum with the real-axis response, consider a path
through the two-parameter plane,
\begin{equation}
    q(s)
    =
    \varEP{q}+s\widehat{v},
    \qquad
    s\in\bbR,
    \label{eq:ep_crossing_path}
\end{equation}
where $\widehat{v}$ is a fixed direction in parameter space.  Along this path
the two QNM branches exhibit a square-root splitting and their separate
projectors become ill-conditioned and may acquire large operator norms, even
though each projector is invariant under reciprocal rescaling of its left and
right eigenvectors.

In the following proposition, $\omega$ and hence $z_{\omega}$ are held fixed,
while $s$ changes the operator.  The assumptions ensure that the evaluation
point never collides with a pole and that the same rank-two cluster remains
isolated.  Under these conditions, the square-root behavior belongs to a
particular labeling of the two eigenvalues and projectors; it need not appear
in the cluster matrix element evaluated at $z_{\omega}$.

\begin{proposition}[Absence of a real-axis EP branch singularity]
Suppose that $H_{\ell,\theta}(q(s))$ is an analytic family near $s=0$, the
rank-two cluster remains isolated, the complementary resolvent is regular, and
no pole reaches $z_{\omega}$ for a fixed real $\omega>0$.  Suppose also that
the localized channel vectors in Eq.~\eqref{eq:pair_channel_moments} define
analytic matrix elements in one common ECS or continuation domain.  Then
$\mathcal T_{\ell,\mathrm{pair}}(\omega;q(s))$ and the full transmission
amplitude are single-valued and real-analytic in $s$ near the EP.  The greybody
factor $\Gamma_{\ell}(\omega;q(s))$ is consequently real-analytic there.
\end{proposition}

\begin{proof}
For a fixed cluster contour, the moments
$\vsub{P}{pair}^{\theta}$ and $\vsub{\mathcal M}{pair}^{\theta}$ are analytic operator
families.  Hence $\vsub{E}{c}$, $\vsub{K}{pair}^{\theta}$, $\Delta^2$, and the
channel moments in Eq.~\eqref{eq:pair_channel_moments} are single-valued and
analytic.  The denominator of Eq.~\eqref{eq:exact_pair_transmission} is
nonzero by assumption.  Adding the regular complementary amplitude preserves
analyticity.  Restriction to real $s$ and multiplication by the complex
conjugate give the stated result for $|\mathcal T_{\ell}|^2$.
\end{proof}

This result is local and does not imply that the greybody factor changes only
slightly: it may vary rapidly while remaining nonsingular.  The numerical
workflow in Appendix~\ref{app:numerical_realization} can quantify that
variation by tracking the two coalescing
eigenvalues, the individual and cluster projectors, the full greybody factor,
and the three terms in Eq.~\eqref{eq:greybody_pair_rest_interference} along the
same path.  Comparison with direct Jost integration and the flux identity in
Eq.~\eqref{eq:greybody_unitarity} would test the channel reconstruction, while
changes of scaling angle, basis size, matching surfaces, and Riesz contour
would test the implementation of the pair/rest decomposition.

\paragraph{Modal singularity versus observable regularity}
The distinction established by the proposition is important.  The EP is a
genuine singularity of the analytically continued resolvent in the spectral
variable: at $z=\varEP{E}$, the two simple poles coalesce into the double pole
in Eq.~\eqref{eq:pair_resolvent_ep_from_moments}.  The same coalescence makes
the individual projectors in
Eq.~\eqref{eq:individual_projectors_from_cluster} branch dependent and allows
opposite contributions of order $1/\Delta$.  It does not follow, however, that
the resolvent evaluated at a fixed physical frequency is singular as a
function of the operator parameters.

The cancellation is structural rather than a fitted cancellation between two
large modal amplitudes.  Equation~\eqref{eq:exact_pair_transmission} depends
only on the single-valued cluster quantities
$\vsub{E}{c}$, $\vsub{P}{pair}^{\theta}$,
$\vsub{K}{pair}^{\theta}$, and $\Delta^2$.  Consequently, branch exchange and
the growth of individual QNM residues diagnose a singular modal coordinate
system, not a necessary singularity of the physical transmission amplitude.
Under the assumptions of the proposition, the full transmission amplitude
and the greybody factor remain real-analytic across the EP.

This regularity does not imply a small response.  The amplitude may vary
rapidly or become strongly enhanced when the complex pole lies close to the
physical axis, and pair/rest interference may produce pronounced spectral
features.  What is excluded is only a branch singularity caused by the
coalescence and relabeling of the two QNMs.

\textit{Thus, a singular reorganization of the QNM decomposition does not, by
itself, imply a singular real-frequency scattering observable.}

The same $\vsub{R}{pair}$ also acts on a localized source.  Hence the
pair-resolved scattering amplitude in
Eq.~\eqref{eq:exact_pair_transmission} and the driven response analyzed
below are two channel choices for one EP-resolved operator, rather than
independent constructions.

\section{Response to a localized source}
\label{sec:source_response}


We use a Gaussian profile as the simplest controlled model of a localized
perturbation,
\begin{equation}
    f_{\ell}(x;\vsub{x}{s},\sigma)
    =
    \frac{A_{\ell}}{\sqrt{2\pi}\,\sigma}
    \exp\!\left[
        -\frac{(x-\vsub{x}{s})^2}{2\sigma^2}
    \right].
    \label{eq:gaussian_source_profile}
\end{equation}
Here $\vsub{x}{s}$ specifies the source location in tortoise coordinates,
$\sigma$ controls its width, and $A_{\ell}$ is the integrated impulse in
the chosen master-variable normalization.

The physical meaning of the time dependence in
$S_\ell(t)$~\eqref{eq:impulsive_source} follows directly by integrating the wave
equation across $t=0$.  Since
\begin{equation}
    \partial_t^2\psi_{\ell}(t,x)
    +H_{\ell}\psi_{\ell}(t,x)
    =
    -\delta(t)f_{\ell}(x),
\end{equation}
the field is continuous and its first time derivative satisfies
\begin{equation}
    \partial_t\psi_{\ell}(0^+,x)
    -
    \partial_t\psi_{\ell}(0^-,x)
    =
    -f_{\ell}(x).
    \label{eq:impulse_jump_condition}
\end{equation}
For a causal solution that vanishes before the impulse, for simplicity, this is
equivalent to the initial data
\begin{equation}
    \psi_{\ell}(0^+,x)=0,
    \qquad
    \partial_t\psi_{\ell}(0^+,x)=-f_{\ell}(x).
    \label{eq:equivalent_initial_data}
\end{equation}
Thus the model describes an instantaneous, spatially localized kick,
followed by free propagation and ringdown.  The minus sign is fixed by
the sign convention in the original wave equation~\eqref{eq:rwz_time_domain}.

The three ingredients of the driven problem should be distinguished.  The
background and any potential deformation determine $H_{\ell}$ and hence the
resolvent poles.  The profile $f_{\ell}$ determines the initial kick and thus
which part of the spectrum is excited.  An observation functional $g$ then
specifies which component of the propagated field is extracted.  Changing
$f_{\ell}$ or $g$ changes the measured amplitude but leaves the QNM
frequencies unchanged.  The source--observer coefficients introduced below
are therefore bilinear channel quantities rather than intrinsic properties of
a QNM alone.

The frequency-domain field at an observation point is
\begin{equation}
    \widetilde{\psi}_{\ell}(\omega,x)
    =
    \int_{-\infty}^{\infty}
    \rmd x'\,
    G_{\ell}(\omega;x,x')
    f_{\ell}(x').
    \label{eq:source_green_convolution}
\end{equation}
In the complex-scaled basis, it takes the form
\begin{equation}
    \widetilde{\psi}_{\ell}(\omega,x)
    =
    \sum_{\nu}
    \frac{
        \Phi_{\ell\nu}^{R}(x)
        \mathcal C_{\ell\nu}[f]
    }{
        \omega^2-E_{\ell\nu}^{\theta}
    },
    \label{eq:source_spectral_response}
\end{equation}
where
\begin{align}
    \Phi_{\ell\nu}^{R}(x)
    &=
    \braket{x|U_{\theta}^{-1}|\psi_{\ell\nu,\theta}^{R}},
    &
    \mathcal C_{\ell\nu}[f]
    &=
    \braket{\psi_{\ell\nu,\theta}^{L}|U_{\theta}|f_{\ell}}.
    \label{eq:source_and_profile_factors}
\end{align}
The first factor gives the back-rotated coordinate-space profile of the right spectral state evaluated at $x$, while the second is the corresponding left-state overlap with the transformed source and therefore controls the source coupling to that spectral component. Their product is invariant under reciprocal rescalings of the left and right
eigenvectors.

For pointwise extraction one may regard $g$ as the coordinate-evaluation
functional $\bra{x}$, with the usual distributional interpretation.  In a
regularized calculation it may instead be a localized test profile, while for
scattering it is an asymptotic channel functional.  Thus $g$ need not be a
normalizable state or an additional dynamical degree of freedom.  It is simply
the linear map that turns the propagated field into the scalar response under
consideration.  Correspondingly, $\Phi_{\ell\nu}^{R}(x)$ is a coordinate-space
profile evaluated at $x$; only its product with the source overlap contributes
to the response.

For an isolated resonance with $E_{\ell n}^{\theta}=\omega_{\ell n}^2$, the
residue in the $\omega$ plane contains an additional Jacobian from
$z=\omega^2$.  For a general observation functional $g$, we define the
simple-pole excitation coefficient
\begin{equation}
    \mathcal B_{\ell n}[g,f]
    =
    \frac{
        \braket{g|U_{\theta}^{-1}|\psi_{\ell n,\theta}^{R}}
        \braket{\psi_{\ell n,\theta}^{L}|U_{\theta}|f_{\ell}}
    }{2\omega_{\ell n}}.
    \label{eq:simple_pole_excitation_coefficient}
\end{equation}
With the Fourier convention in Eq.~\eqref{eq:fourier_pair}, the corresponding
positive-frequency contribution to the causal signal is proportional to
$-\rmi\mathcal B_{\ell n}\rme^{-\rmi\omega_{\ell n}t}$.  Equation~\eqref{eq:simple_pole_excitation_coefficient}
separates the universal pole location from the source- and profile-dependent
amplitude.

Related notions of QNM excitation coefficients and source-independent
excitation factors have long been developed in modal and Green-function
analyses
\cite{Nollert:1998ys,Berti:2006wq}.
The coefficient in Eq.~\eqref{eq:simple_pole_excitation_coefficient},
however, is explicitly a source-to-observation quantity and therefore
contains both the prescribed source overlap and the observation functional.

For diagnostic purposes, the finite-basis sum can be divided into
isolated resonance eigenvalues and pseudostates associated with the
rotated continuum,
\begin{equation}
    \widetilde{\psi}_{\ell}
    =
    \vsup{\widetilde{\psi}}{pole}_{\ell}
    +
    \vsup{\widetilde{\psi}}{cont}_{\ell} .
    \label{eq:pole_continuum_split}
\end{equation}
Only the total is independent of the numerical representation. The separate pieces must be checked under changes of the scaling angle and basis size. These checks are especially important when a broad resonance approaches the rotated continuum. Scanning $\vsub{x}{s}$ and $\sigma$ in Eq.~\eqref{eq:gaussian_source_profile} then identifies source profiles that preferentially excite a given pole or suppress its overlap. In the narrow-source limit, $\mathcal C_{\ell\nu}[f]$ probes the local structure of the left spectral state, but the observable amplitude still requires its product with the corresponding right-state profile $\Phi_{\ell\nu}^{R}(x)$ in Eq.~\eqref{eq:source_and_profile_factors}.

The time-domain signal is obtained from
\begin{equation}
    \psi_{\ell}(t,x)
    =
    \frac{1}{2\pi}
    \int_{\mathcal C}
    \rmd\omega\,
    \rme^{-\rmi\omega t}
    \widetilde{\psi}_{\ell}(\omega,x).
    \label{eq:source_inverse_fourier}
\end{equation}
Deforming $\mathcal C$ into the lower half-plane separates residues of
simple QNM poles from the deformed-continuum integral.  The latter is
essential for the prompt response and, in the asymptotically flat case,
for the branch-cut contribution associated with the late-time tail.\footnote{%
A recent decomposition of the Schwarzschild Green function makes the
importance of the non-pole sector particularly explicit
\cite{Su:2026fvj}.
By separating the frequency-domain kernel into two components with
different large-frequency behavior, the direct response, the QNM
contribution, and the late-time tail can be isolated using contours
adapted to different causal spacetime regions.
This analysis also shows that a QNM residue expansion is not a globally
valid representation of the response at all times: its applicable
domain depends on the source and observation locations.}
A finite set of complex-scaled pseudostates approximates this integral
over a finite frequency and time window; convergence of the waveform is
therefore a stronger test than convergence of isolated eigenvalues.

A physical particle or matter source can be included by replacing the
Gaussian profile with the appropriate harmonic projection
$\widetilde S_{\ell m}^{\parity}(\omega,x)$.  The resolvent of the chosen background operator is unchanged.
A point-particle source generally contains radial delta distributions
and their derivatives and therefore does not define an ordinary
dilation-analytic source vector.  Its transformation under global
complex scaling requires an additional prescription.  If the radial
support of the source is contained within the undeformed region of an
ECS contour, however, the source remains on the real axis and can be
imposed through the usual jump conditions.  We leave a detailed
implementation of particle sources for future work.


\section{Exceptional-point limit and generalized source amplitudes}
\label{sec:exceptional_points}

\subsection{Gaussian-deformed Regge--Wheeler model setting}
\label{sec:gaussian_bump_ep}

As a concrete two-parameter family containing an exceptional point,
we follow the environmental-potential model of Ref.~\cite{Yang:2025dbn}.  We modify the odd-parity $\ell=2$ operator according to
\begin{align}
    H_{\varepsilon,d,\sigma}
    & =
    -\frac{\rmd^2}{\rmd x^2}
    +\vsup{V}{odd}_{2}(x)
    +\vsub{V}{G}(x;\varepsilon,d,\sigma),
    \label{eq:ep_deformed_operator}
    \\
    \vsub{V}{G}(x;\varepsilon,d,\sigma)
    & =
    \frac{\varepsilon}{M^2}
    \exp\!\left[-\frac{(x-d)^2}{2\sigma^2}\right].
    \label{eq:ep_gaussian_potential}
\end{align}
Here $\varepsilon$ is dimensionless, while $d$ and $\sigma$ specify the
location and width of the additional potential in the tortoise coordinate. The
factor $M^{-2}$ makes its dimensions identical to those of the
Regge--Wheeler potential.  The model is phenomenological: it is used to deform
the spectral operator and tune its resonances. It must not be confused with
the inhomogeneous source $S_{\ell m}^{\parity}$ on the right-hand side of
Eq.~\eqref{eq:rwz_time_domain}.  In the driven problem, the Gaussian potential
sets the resolvent, whereas $f_{\ell}$ in Eq.~\eqref{eq:impulsive_source}
selects how that resolvent is excited.

Reference~\cite{Yang:2025dbn} adopts $M=1$ and the tortoise-coordinate
convention
\begin{equation}
    \vsub{x}{Y}
    =
    r+2\ln(r-2),
    \label{eq:yang_tortoise}
\end{equation}
with $\vsub{V}{G}=\varepsilon\exp[-(\vsub{x}{Y}-\vsub{d}{Y})^2]$.
This corresponds to $2\sigma^2=1$ in Eq.~\eqref{eq:ep_gaussian_potential}.
For this family, the fundamental mode and first overtone coalesce at the
reported values
\begin{align}
    \varEP{\varepsilon}
    &\simeq
    10^{-2.294}
    \simeq
    5.08\times10^{-3},
    \label{eq:yang_ep_benchmark_epsilon}
    \\
    \vsub{d}{Y,EP}
    &\simeq
    15.698,
    \\
    \varEP{\omega}
    &\simeq
    0.365-0.117\rmi.
    \label{eq:yang_ep_benchmark}
\end{align}
The numerical value of a position in the tortoise coordinates depends on its
additive constant.  Our convention in Eq.~\eqref{eq:tortoise_coordinate}, with
$M=1$, obeys
\begin{equation}
    \vsub{x}{Y}
    =
    x+2\ln2.
\end{equation}
The same physical Gaussian center is therefore
\begin{equation}
    \varEP{d}
    =
    \vsub{d}{Y,EP}-2\ln2
    \simeq
    14.312
    \label{eq:converted_ep_center}
\end{equation}
in our convention.  Stating this shift is necessary when the model is
implemented using a different integration constant for $r_*$.  The benchmark values in Eqs.~\eqref{eq:yang_ep_benchmark_epsilon}--\eqref{eq:yang_ep_benchmark}, together with the converted center in Eq.~\eqref{eq:converted_ep_center}, supply a concrete target for a future finite-basis realization of the
contour-moment construction.

An EP has real codimension two in a generic non-Hermitian family, so the two
controls $(\varepsilon,d)$ are essential.  Encircling the reported point in
this parameter plane exchanges the two local QNM branches after one circuit;
they return to their original labels only after two circuits.  The
single-valued cluster moments in Sec.~\ref{sec:rank_two_riesz_cluster} do not
require a branch label and are therefore the natural variables for describing
this monodromy.  In a finite-basis implementation, the EP location and cluster
moments should also be stable under changes of the scaling angle and basis
parameters.

\subsection{Failure of the simple-pole decomposition}
\label{sec:ep_laurent}

Let $q=(\varepsilon,d)$ and consider the complex-scaled operator
$H_{\ell,\theta}(q)$.  At a second-order exceptional point, two eigenvalues and
their eigenvectors coalesce,
\begin{equation}
    E_+(\varEP{q})
    =
    E_-(\varEP{q})
    =
    \varEP{E},
    \qquad
    \varEP{E}=\varEP{\omega}^2,
    \label{eq:ep_eigenvalue_coalescence}
\end{equation}
and the eigenspace has geometric multiplicity one.  A right Jordan chain is
defined by
\begin{align}
    \left(H_{\ell,\theta}-\varEP{E}\right)
    \ket{\chi_0^R}
    &=0,
    \label{eq:right_jordan_chain_zero}
    \\
    \left(H_{\ell,\theta}-\varEP{E}\right)
    \ket{\chi_1^R}
    &=
    \ket{\chi_0^R},
    \label{eq:right_jordan_chain}
\end{align}
with an analogous left chain.  The analytically continued resolvent remains a well-defined
operator-valued meromorphic function of $z$ near $z=\varEP{E}$.  What fails at the EP is its
representation as a sum of two independent simple-pole projectors.  

To see explicitly how the double pole follows from the Jordan chain, define
the two-dimensional root space
\begin{equation}
    \vsub{\mathcal X}{EP}
    :=
    \operatorname{span}
    \left\{
        \ket{\chi_0^R},
        \ket{\chi_1^R}
    \right\},
\end{equation}
and let $\vsub{I}{EP}$ denote the identity on this space.  The nilpotent
part of the restricted operator is
\begin{equation}
    \vsub{N}{EP}
    :=
    \left.
    \left(
        H_{\ell,\theta}-\varEP{E}
    \right)
    \right|_{\vsub{\mathcal X}{EP}}.
\end{equation}
Equations~\eqref{eq:right_jordan_chain_zero} and
\eqref{eq:right_jordan_chain} imply
\begin{align}
    \vsub{N}{EP}\ket{\chi_0^R}
    &=0,
    &
    \vsub{N}{EP}\ket{\chi_1^R}
    &=
    \ket{\chi_0^R},
\end{align}
and hence
\begin{equation}
    \vsub{N}{EP}^{\,2}=0,
    \qquad
    \vsub{N}{EP}\neq0.
\end{equation}
The resolvent restricted to the root space is therefore
\begin{align}
    \left.
    R_{\ell,\theta}(z)
    \right|_{\vsub{\mathcal X}{EP}}
    &=
    \left[
        (z-\varEP{E})\vsub{I}{EP}
        -\vsub{N}{EP}
    \right]^{-1}
    \\
    &=
    \frac{\vsub{I}{EP}}{z-\varEP{E}}
    +
    \frac{\vsub{N}{EP}}
         {(z-\varEP{E})^2}.
    \label{eq:jordan_resolvent_on_root_space}
\end{align}
The expansion terminates because $\vsub{N}{EP}^{\,2}=0$.  Since the
spectral complement contributes a term analytic at $\varEP{E}$, the full
resolvent consequently has the local Laurent form
\begin{equation}
    R_{\ell,\theta}(z)
    =
    \frac{\Laurent_{-2}}{(z-\varEP{E})^2}
    +
    \frac{\Laurent_{-1}}{z-\varEP{E}}
    +
    \vsub{R}{reg}(z),
    \label{eq:ep_laurent_expansion}
\end{equation}
where $\vsub{R}{reg}$ is analytic at $\varEP{E}$.  The operators
$\Laurent_{-2}$ and $\Laurent_{-1}$ are finite and are independent of an
arbitrary rescaling of the Jordan vectors.

The failure here is a failure of modal coordinates, not of the resolvent away
from its pole.  Off the EP, the two simple eigenvectors provide a basis of the
root space and allow two residues to be quoted separately.  At the EP those
vectors become parallel, so extracting two coefficients in that basis becomes
an ill-conditioned operation.  The root space itself remains two-dimensional
in the algebraic sense, and the two Laurent operators give a finite,
basis-independent description of the operator on that space.

The same conclusion can be seen by approaching the EP along an analytic path
$q(s)$ with $q(0)=\varEP{q}$ and a generic tangent direction.  The two
eigenvalues possess the Puiseux expansions familiar from non-Hermitian
spectral theory~\cite{Moiseyev:2011},
\begin{equation}
    E_{\pm}
    =
    \varEP{E}
    \pm c\sqrt{s}
    +O(s),
    \label{eq:ep_puiseux_eigenvalues}
\end{equation}
while the two separately normalized residues can contain opposite terms of
order $s^{-1/2}$.  Their sum remains finite and tends to
Eq.~\eqref{eq:ep_laurent_expansion}.  This cancellation is displayed directly
by Eq.~\eqref{eq:individual_projectors_from_cluster}.  It is therefore the
cluster moments, rather than either separately normalized residue, that extend
analytically through the EP.  The same observation also motivates a numerical
implementation based on contour moments instead of the subtraction of two
large, normalization-sensitive modal contributions.

\subsection{Jordan-chain interpretation of the Riesz moments}
\label{sec:ep_riesz_projector}

We now relate the moment representation in
Sec.~\ref{sec:rank_two_riesz_cluster} to the conventional Jordan-chain
description.  Following the spectral projection formalism of
Kato~\cite{Kato:1976}, let
$\vsub{\Gamma}{pair}$ be a positively oriented closed contour in the
$z$ plane that encloses the two QNM eigenvalues that form the EP and no other
spectral point of $H_{\ell,\theta}$.  The associated Riesz projector is
the operator~$\vsub{P}{pair}^\theta$ introduced in
Eq.~\eqref{eq:ep_riesz_projector}.  Here
$R_{\ell,\theta}(z)=(z-H_{\ell,\theta})^{-1}$, so that definition agrees
with the standard spectral projection of Kato.
As long as the contour remains in the resolvent set and continues to isolate
the same spectral cluster, $\vsub{P}{pair}^{\theta}$ projects onto the
invariant root subspace associated with the enclosed eigenvalues.  For an
analytic parameter dependence of the operator, this total projector varies
regularly even when the individual eigenvalues and eigenvectors develop the
branch-point structure characteristic of an EP~\cite{Kato:1976}.

Away from the EP, where the two enclosed eigenvalues are simple and the
operator is diagonalizable in this two-mode sector, the Riesz projector
reduces to the sum of the ordinary biorthogonal spectral projectors,
\begin{align}
    \vsub{P}{pair}^{\theta}
    & =
    P_{1}^{\theta}+P_{2}^{\theta},
    \\
    P_{j}^{\theta}
    & =
    \frac{
        \ket{\psi^R_{j,\theta}}\bra{\psi^L_{j,\theta}}
    }{
        \braket{\psi^L_{j,\theta}|\psi^R_{j,\theta}}
    },
    \qquad
    j=1,2.
    \label{eq:ep_pair_projector_sum}
\end{align}
With the biorthogonal convention $\braket{\psi^L_{i,\theta}|\psi^R_{j,\theta}}=\delta_{ij}$, the
denominator in the second line is unity.  Close to the EP, the two projectors
$P_{1}^{\theta}$ and $P_{2}^{\theta}$ are invariant under reciprocal
rescalings of their left and right eigenvectors, but their construction becomes
ill-conditioned and their operator norms may grow without bound.  Their sum
remains the well-defined projector onto the combined spectral subspace.
This cancellation is the operator-level counterpart of the non-divergent
Jordan-block representations used for non-Hermitian systems near EPs~\cite{Hashimoto:2014ep}.

At the EP itself, Eq.~\eqref{eq:ep_pair_projector_sum} must not be interpreted
as a sum of two independent eigenvector projectors.  The eigenspace has
geometric multiplicity one, whereas the algebraic multiplicity of the
coalescing cluster is two.  Accordingly,
$\vsub{P}{pair}^{\theta}$ projects onto the two-dimensional root space
spanned by a Jordan chain,
\begin{equation}
    \operatorname{im}(\vsub{P}{pair}^{\theta})
    =
    \operatorname{span}
    \left\{
        \ket{\chi_0^R},
        \ket{\chi_1^R}
    \right\},
    \label{eq:ep_riesz_root_space}
\end{equation}
where the vectors satisfy Eqs.~\eqref{eq:right_jordan_chain_zero} and \eqref{eq:right_jordan_chain}. Thus the rank of
the Riesz projector continues to count the algebraic multiplicity of the
isolated spectral cluster even though the number of linearly independent
eigenvectors has dropped to one.

For the source problem, the stable quantity corresponding to the combined
``mode-1 plus mode-2'' component is therefore the pair-projected source,
\begin{equation}
    \ket{\vsub{f}{pair}^{\theta}}
    :=
    \vsub{P}{pair}^{\theta}
    U_{\theta}\ket{f}.
    \label{eq:ep_pair_projected_source}
\end{equation}
A frequency-dependent response, however, requires slightly more information
than the projector alone.  The pair resolvent~$\vsub{R}{pair}^\theta$ introduced in
Eq.~\eqref{eq:ep_pair_resolvent} is the restriction of the full resolvent to
the same invariant subspace.  Since the Riesz projector commutes with the
operator, it also obeys
\begin{equation}
    \vsub{R}{pair}^{\theta}(z)
    =
    \vsub{P}{pair}^{\theta}
    R_{\ell,\theta}(z)
    \vsub{P}{pair}^{\theta}.
\end{equation}
Away from the EP this gives
\begin{equation}
    \vsub{R}{pair}^{\theta}(z)
    =
    \frac{P_1^{\theta}}{z-E_1}
    +
    \frac{P_2^{\theta}}{z-E_2},
    \label{eq:ep_pair_resolvent_simple}
\end{equation}
which is precisely the two-mode contribution of the spectral representation
used in Sec.~\ref{sec:complex_scaling}.  At the EP it has the finite Jordan
form
\begin{equation}
    \vsub{R}{pair}^{\theta}(z)
    =
    \frac{\varEP{N}^{\theta}}
         {(z-\varEP{E})^2}
    +
    \frac{\vsub{P}{pair}^{\theta}}
         {z-\varEP{E}},
    \label{eq:ep_pair_resolvent_jordan}
\end{equation}
where
\begin{equation}
    \varEP{N}^{\theta}
    =
    \left(H_{\ell,\theta}-\varEP{E}\right)
    \vsub{P}{pair}^{\theta}
    =
    \vsub{K}{pair}^{\theta}(\varEP{q}),
    \qquad
    \left(\varEP{N}^{\theta}\right)^2=0
    \label{eq:ep_nilpotent_operator}
\end{equation}
within a second-order Jordan block.
Here $\vsub{N}{EP}^{\theta}$ is the full-space embedding of the
restricted nilpotent operator $\vsub{N}{EP}$.  Its restriction to
$\vsub{\mathcal X}{EP}$ agrees with $\vsub{N}{EP}$, while it vanishes on
the complementary Riesz subspace.
Comparison with
Eq.~\eqref{eq:ep_laurent_expansion} identifies
\begin{equation}
    \Laurent_{-1}
    =
    \vsub{P}{pair}^{\theta},
    \qquad
    \Laurent_{-2}
    =
    \varEP{N}^{\theta}.
    \label{eq:ep_riesz_laurent_identification}
\end{equation}
Thus the Jordan form is not an additional construction: it is the EP limit of
the exact two-moment representation.  It avoids assigning separate amplitudes
to the coalescing modes and permits the source-projection analysis of
Sec.~\ref{sec:complex_scaling} to pass through the EP using one invariant
spectral cluster.  In a numerical realization, the same formulation also
avoids resolving two nearly parallel eigenvectors.

\subsection{Generalized amplitudes for a driven response}
\label{sec:ep_generalized_amplitudes}

For a source $f$ and an observation functional $g$, the physical matrix
elements of the two cluster operators are
\begin{align}
    p_{g,f}(q)
    &=
    \braket{
        g
        |
        U_{\theta}^{-1}
        \vsub{P}{pair}^{\theta}(q)
        U_{\theta}
        |
        f
    },
    \\
    k_{g,f}(q)
    &=
    \braket{
        g
        |
        U_{\theta}^{-1}
        \vsub{K}{pair}^{\theta}(q)
        U_{\theta}
        |
        f
    }.
    \label{eq:physical_cluster_matrix_elements}
\end{align}
They are single-valued cluster quantities and do not require the two QNM
branches to be labeled separately.  Proposition~\ref{prop:exact_rank_two_cluster}
then gives
\begin{equation}
    F_{\mathrm{pair}}(z;q)
    =
    \frac{
        [z-\vsub{E}{c}(q)]p_{g,f}(q)
        +
        k_{g,f}(q)
    }{
        [z-\vsub{E}{c}(q)]^{2}
        -
        \Delta(q)^{2}
    }.
    \label{eq:exact_pair_source_response}
\end{equation}
This formula replaces two individually ill-conditioned QNM amplitudes by two
single-valued cluster amplitudes.

At the EP, the two cluster matrix elements become the Laurent coefficients
of the physical response:
\begin{align}
    C_{-1}[g,f]
    &:=
    p_{g,f}(\varEP{q})
    =
    \braket{
        g
        |
        U_{\theta}^{-1}
        \Laurent_{-1}
        U_{\theta}
        |
        f
    },
    \\
    C_{-2}[g,f]
    &:=
    k_{g,f}(\varEP{q})
    =
    \braket{
        g
        |
        U_{\theta}^{-1}
        \Laurent_{-2}
        U_{\theta}
        |
        f
    }.
    \label{eq:generalized_ep_amplitudes}
\end{align}
The singular part of the full physical response is consequently
\begin{equation}
    \braket{g|R_{\ell}(z)|f}
    =
    \frac{C_{-2}[g,f]}{(z-\varEP{E})^2}
    +
    \frac{C_{-1}[g,f]}{z-\varEP{E}}
    +
    \text{regular terms}.
    \label{eq:ep_source_response}
\end{equation}
The coefficients depend on both the source and the observation channel, are
independent of a Jordan-chain normalization, and remain finite at the
coalescence.

There are four related but distinct levels of description.  The operators
$\Laurent_{-2}$ and $\Laurent_{-1}$ are Laurent coefficients of the resolvent
in the $z$ plane.  Their matrix elements are the channel coefficients
$C_{-2}[g,f]$ and $C_{-1}[g,f]$.  Conversion from $z$ to $\omega$ gives
$A_{-2}[g,f]$ and $A_{-1}[g,f]$, and the retarded inverse Fourier transform
then gives the time-domain coefficients $D_1$ and $D_0$.  Each step is fixed
once $H_{\ell}$, $f$, $g$, and the Fourier convention are specified; none of
these quantities is introduced as an independent fit parameter.

For a second-order EP the nonzero nilpotent operator has rank one.  With
compatible left and right Jordan-chain normalizations it may be written as
\begin{equation}
    \varEP{N}^{\theta}
    =
    \ket{\chi_0^R}\bra{\chi_0^L},
\end{equation}
and hence
\begin{equation}
    C_{-2}[g,f]
    =
    \braket{g|U_{\theta}^{-1}|\chi_0^R}
    \braket{\chi_0^L|U_{\theta}|f}.
    \label{eq:double_pole_channel_factorization}
\end{equation}
Although the two factors depend on the reciprocal normalization of the Jordan
vectors, their product does not.  Equation~\eqref{eq:double_pole_channel_factorization}
also makes the selection rule transparent: the double-pole term is absent if
either the source has zero overlap with the left EP eigenvector or the
observation functional annihilates the right EP eigenvector.

In particular,
$C_{-2}[g,f]=0$ means that the leading double-pole contribution is
absent from the specified source-to-observation channel.  Consequently,
the corresponding time-domain response contains no term proportional
to $t\rme^{-\rmi\varEP{\omega}t}$, although a simple-pole
contribution proportional to
$C_{-1}[g,f]\rme^{-\rmi\varEP{\omega}t}$ may still remain.

On the other hand, consider the case that $C_{-2}[g,f]\neq0$.
Because the evolution is Fourier transformed in $\omega$ rather than $z$, it is useful to display the local conversion. With
$z=\omega^2$, $\varEP{E}=\varEP{\omega}^2$ and
$\delta\omega=\omega-\varEP{\omega}$,
\begin{align}
    \frac{1}{(z-\varEP{E})^2}
    & =
    \frac{1}{4\varEP{\omega}^2}
    \frac{1}{(\delta\omega)^2}
    -
    \frac{1}{4\varEP{\omega}^3}
    \frac{1}{\delta\omega}
    +O(1),
    \\
    \frac{1}{z-\varEP{E}}
    & =
    \frac{1}{2\varEP{\omega}}
    \frac{1}{\delta\omega}
    +O(1).
    \label{eq:z_to_omega_ep_conversion}
\end{align}
Thus the two singular coefficients in the $\omega$ plane are
\begin{align}
    A_{-2}[g,f]
    &=
    \frac{C_{-2}[g,f]}{4\varEP{\omega}^{2}},
    \\
    A_{-1}[g,f]
    &=
    \frac{C_{-1}[g,f]}{2\varEP{\omega}}
    -
    \frac{C_{-2}[g,f]}{4\varEP{\omega}^{3}}.
    \label{eq:omega_plane_ep_coefficients}
\end{align}
Closing the retarded Fourier contour in the lower half plane gives the causal
positive-frequency contribution
\begin{equation}
    \varEP{\psi}(t)
    =
    \Theta(t)
    \left(D_0+D_1t\right)
    \rme^{-\rmi\varEP{\omega}t},
    \label{eq:ep_time_dependence}
\end{equation}
where the Fourier convention in Eq.~\eqref{eq:fourier_pair} fixes
\begin{align}
    D_0
    &=
    -\rmi A_{-1}[g,f]
    =
    -\frac{\rmi C_{-1}[g,f]}{2\varEP{\omega}}
    +
    \frac{\rmi C_{-2}[g,f]}{4\varEP{\omega}^{3}},
    \\
    D_1
    &=
    -A_{-2}[g,f]
    =
    -\frac{C_{-2}[g,f]}{4\varEP{\omega}^{2}}.
    \label{eq:time_domain_ep_coefficients}
\end{align}

The origin of the polynomial factor can be seen directly from the residue of
a second-order pole,
\begin{equation}
    \Res_{\omega=\varEP{\omega}}
    \frac{
        \rme^{-\rmi\omega t}
    }{
        (\omega-\varEP{\omega})^2
    }
    =
    -\rmi t\,
    \rme^{-\rmi\varEP{\omega}t}.
    \label{eq:double_pole_fourier_residue}
\end{equation}
The derivative required by the second-order residue produces the factor of
$t$, while closing the retarded contour in the lower half plane supplies the
remaining contour factor and yields $D_1=-A_{-2}$.  The term
$t\rme^{-\rmi\varEP{\omega}t}$ is therefore forced by the double pole rather
than appended to the waveform by hand.

The symmetry-related contribution at $-\varEP{\omega}^*$ must be added when a
real field is reconstructed.  This polynomial modulation is the finite EP limit of the
interference between the two nearby resonances
\cite{Yang:2025dbn,PanossoMacedo:2025xnf,Cheng:2026ep}.
A complementary semi-analytic treatment of nearly double-pole QNMs in the
Nariai limit finds the same transient linear growth and determines
conditions under which it dominates the early ringdown
\cite{Nakamoto:2026lyo}.

The Riesz-moment formulation gives a direct contour prescription for the
Laurent coefficients without constructing an explicitly normalized Jordan
basis.  The coefficient of the simple-pole term is the Riesz projector,
\begin{equation}
    \Laurent_{-1}
    =
    \vsub{P}{pair}^{\theta},
\end{equation}
while the double-pole coefficient requires the first spectral moment as well,
\begin{align}
    \Laurent_{-2}
    &=
    \vsub{\mathcal M}{pair}^{\theta}
    -\varEP{E}\vsub{P}{pair}^{\theta}
    \notag\\
    &=
    \frac{1}{2\pi\rmi}
    \oint_{\vsub{\Gamma}{pair}}
    (z-\varEP{E})
    R_{\ell,\theta}(z)\,\rmd z.
    \label{eq:ep_contour_coefficients}
\end{align}
Thus the projector alone determines $\Laurent_{-1}$, whereas the pair
$\left(\vsub{P}{pair}^{\theta},\vsub{\mathcal M}{pair}^{\theta}\right)$ determines both
Laurent operators and the full rank-two pair resolvent.
In practice, the contour operations may be applied directly to the transformed
source vector.  This yields $\vsub{P}{pair}^{\theta}U_{\theta}f$ and
$\Laurent_{-2}U_{\theta}f$ without resolving two nearly parallel
eigenvectors.  The workflow in Appendix~\ref{app:numerical_realization} is
designed to compare the contour result with the combined two-pole response
away from the EP and verify the smooth limit as
$(\varepsilon,d)$ approaches the reported point.

The Gaussian-potential model also connects the EP analysis to the other
observables considered here.  The pair resolvent used for the scattering
decomposition in Sec.~\ref{sec:pair_resolved_transmission} also determines the
response to the localized impulse.  One can therefore distinguish a large
intrinsic resolvent norm near the EP from a large measured signal, which
additionally requires a nonvanishing channel or source--observer projection.

\subsection{Equivalence with the Jost double-zero expansion}
\label{sec:ep_jost_equivalence}

The same coefficients can be obtained without complex scaling by expanding
the Jost representation locally.  For fixed source and observation states,
Eq.~\eqref{eq:jost_green_function} and the integrations defining the matrix
element have the form
\begin{equation}
    F_{g,f}(\omega)
    :=
    \braket{g|R_{\ell}(\omega^2)|f}
    =
    \frac{\mathcal Q_{g,f}(\omega)}{W_{\ell}(\omega)},
    \label{eq:jost_scalar_response}
\end{equation}
where $\mathcal Q_{g,f}$ contains the Jost-solution numerator and the
source--observation integrations.  It is analytic locally when $f$ and $g$
are admissible and no other singularity is present.

This comparison is more than a consistency check between two numerical
procedures.  The Jost construction imposes the physical radiation conditions
directly, whereas the Riesz construction accesses the same meromorphic scalar
response through a contour of the complex-scaled resolvent.  A meromorphic
function has unique Laurent coefficients at an isolated pole.  Agreement of
the two expressions therefore shows that the contour moments recover
coefficients of the physical source--observer response and not quantities
created by the complex deformation.

\begin{proposition}[Jost and Riesz--Laurent coefficients]
\label{prop:jost_riesz_equivalence}
Suppose that $\varEP{\omega}$ is a second-order zero of the Jost Wronskian,
\begin{equation}
    W_{\ell}(\varEP{\omega})=0,
    \qquad
    W_{\ell}'(\varEP{\omega})=0,
    \qquad
    W_{\ell}''(\varEP{\omega})\neq0.
    \label{eq:jost_double_zero_conditions}
\end{equation}
Then the singular part of Eq.~\eqref{eq:jost_scalar_response} is
\begin{equation}
    F_{g,f}(\omega)
    =
    \frac{A_{-2}^{J}[g,f]}
         {(\omega-\varEP{\omega})^2}
    +
    \frac{A_{-1}^{J}[g,f]}
         {\omega-\varEP{\omega}}
    +O(1),
    \label{eq:jost_omega_laurent_expansion}
\end{equation}
where
\begin{align}
    A_{-2}^{J}[g,f]
    &=
    \frac{
        2\mathcal Q_{g,f}(\varEP{\omega})
    }{
        W_{\ell}''(\varEP{\omega})
    },
    \label{eq:jost_double_pole_coefficient}
    \\
    A_{-1}^{J}[g,f]
    &=
    \frac{
        2\mathcal Q_{g,f}'(\varEP{\omega})
    }{
        W_{\ell}''(\varEP{\omega})
    }
    -
    \frac{
        2\mathcal Q_{g,f}(\varEP{\omega})
        W_{\ell}'''(\varEP{\omega})
    }{
        3\left[W_{\ell}''(\varEP{\omega})\right]^2
    }.
    \label{eq:jost_simple_pole_coefficient}
\end{align}
They coincide with the Riesz coefficients through
\begin{align}
    A_{-2}^{J}[g,f]
    &=
    \frac{C_{-2}[g,f]}{4\varEP{\omega}^{2}},
    \label{eq:jost_riesz_coefficient_matching2}
    \\
    A_{-1}^{J}[g,f]
    &=
    \frac{C_{-1}[g,f]}{2\varEP{\omega}}
    -
    \frac{C_{-2}[g,f]}{4\varEP{\omega}^{3}}.
    \label{eq:jost_riesz_coefficient_matching}
\end{align}
\end{proposition}

\begin{proof}
Expand the two analytic functions as
\begin{align}
    \mathcal Q_{g,f}(\omega)
    &=
    \mathcal Q_{g,f}(\varEP{\omega})
    +
    \mathcal Q_{g,f}'(\varEP{\omega})
    \delta\omega
    +O(\delta\omega^2),
    \\
    W_{\ell}(\omega)
    &=
    \frac{1}{2}W_{\ell}''(\varEP{\omega})
    \delta\omega^2
    +
    \frac{1}{6}W_{\ell}'''(\varEP{\omega})
    \delta\omega^3
    +O(\delta\omega^4),
\end{align}
where $\delta\omega=\omega-\varEP{\omega}$.  Division of the two series gives
Eqs.~\eqref{eq:jost_double_pole_coefficient} and
\eqref{eq:jost_simple_pole_coefficient}.  Equations
\eqref{eq:z_to_omega_ep_conversion} and
\eqref{eq:omega_plane_ep_coefficients} then give
Eqs.~\eqref{eq:jost_riesz_coefficient_matching2}and~\eqref{eq:jost_riesz_coefficient_matching} by uniqueness of the Laurent
expansion.
\end{proof}

The Jost formulas and the Riesz contour moments are therefore two exact
representations of the same source--observer coefficients.  The former uses
derivatives at a double zero; the latter packages the result into finite-rank
operators and applies without choosing a separately normalized pair of QNMs.

\section{Conclusion}
\label{sec:conclusion}

The main result of this work is that the zeroth and first Riesz moments of an
isolated rank-two QNM cluster determine its complete resolvent through a
second-order exceptional point.  The exact formula in
Eq.~\eqref{eq:exact_pair_resolvent} replaces two branch-dependent and
individually ill-conditioned QNM projectors by the single-valued cluster data
$\vsub{P}{pair}^{\theta}$, $\vsub{\mathcal M}{pair}^{\theta}$,
$\vsub{E}{c}$, $\vsub{K}{pair}^{\theta}$, and $\Delta^2$.  At the EP, the
two contour moments give the finite Laurent operators
\begin{equation}
    \Laurent_{-1}
    =
    \vsub{P}{pair}^{\theta},
    \qquad
    \Laurent_{-2}
    =
    \vsub{\mathcal M}{pair}^{\theta}
    -\varEP{E}\vsub{P}{pair}^{\theta}.
\end{equation}
The projector determines the simple-pole coefficient, while the first moment
contains the additional nilpotent information required for the double pole.
This construction uses the invariant root subspace and the restriction of the
operator to it; it neither labels the two modes separately nor requires an
explicitly normalized Jordan chain.

The first physical consequence concerns real-frequency scattering.  Channel
matrix elements of the same cluster formula give an exact pair/rest
decomposition of the transmission amplitude, with the interference term
retained in the greybody factor.  We proved that, as a function of the operator
parameters, the fixed-real-frequency transmission amplitude and greybody
factor have no EP branch singularity as long as the cluster remains isolated,
the complementary resolvent is regular, and no pole reaches the physical axis.
They may nevertheless vary rapidly or become strongly enhanced.  The result
therefore distinguishes a singular choice of modal coordinates from a
singular physical observable.

The second consequence concerns a driven response.  Matrix elements of
$\Laurent_{-1}$ and $\Laurent_{-2}$ between a source $f$ and an observation
functional $g$ define the finite coefficients $C_{-1}[g,f]$ and
$C_{-2}[g,f]$.  They determine both the constant and linear-in-time terms in
the causal EP signal.  In particular, the familiar form
$(D_0+D_1t)\rme^{-\rmi\varEP{\omega}t}$ is obtained from the resolvent rather
than postulated as a fitting template, and the absence of its linear term is a
source--observer selection rule expressed by $C_{-2}[g,f]=0$.  The exact
agreement with the double-zero expansion of the Jost Wronskian shows that
these amplitudes are invariants of the physical channel, not artifacts of
complex scaling or free phenomenological parameters.  For dilation-analytic
source and observation states, their values are independent of the scaling
angle in the exact theory.

We use the Gaussian-deformed $\ell=2$ Regge--Wheeler potential of
Ref.~\cite{Yang:2025dbn} as a concrete model setting.\footnote{%
Localized potential deformations, including Gaussian bumps, have also
been used to diagnose the spectral sensitivity of black-hole QNMs
\cite{Cheung:2021bol}.}  The potential bump
changes the operator, whereas the Gaussian impulse or scattering channel
specifies how its root subspace is probed.  Keeping these roles distinct is
essential to a source-dependent interpretation of EP enhancement.

We have deliberately separated the analytic construction from its numerical
realization.  Numerical calculations remain necessary to assign quantitative
values to a chosen black-hole model and can test the contour implementation
against direct Jost integration, flux conservation, and time evolution.  They
do not, however, define the relevant near-EP variables.  The cluster
invariants, the real-axis regularity statement, and the source--observer
Laurent coefficients follow from the spectral structure of the outgoing
resolvent.  In this sense, the present formulation supplies a constructive
spectral foundation on which numerical evaluation and waveform inference can
subsequently be built.

The present closed form exploits the two-dimensionality of the isolated
root subspace and therefore addresses second-order EPs.
A natural extension is to larger isolated root subspaces.
At a higher-order EP, the restricted operator develops a higher-index
nilpotent part and the resolvent correspondingly contains higher-order
Laurent poles.
Black-hole spectra already exhibit nontrivial multi-mode topology even
without a higher-order coalescence: in Kerr--Newman, loops enclosing
multiple pairwise EPs can induce cyclic permutations among three QNM
branches
\cite{Cavalcante:2026vgr}.
A distinct challenge is posed by resonance--continuum coalescence, for
which complex-scaling studies have shown that resonant and scattering
states can coalesce at an EP as the resonance pole becomes embedded in the
continuum
\cite{Morikawa:2025inx}.
Developing a resolvent-based response construction that accommodates both
higher-rank isolated clusters and resonance--continuum coalescence is an
important direction for future work.

\section*{Note added}
After the completion and initial posting of this manuscript, we became aware of the independent and closely contemporaneous work of Ould~El~Hadj and Dolan~\cite{OuldElHadj:2026vym}, which studies exceptional points in the Schwarzschild quasinormal-mode spectrum generated by localized perturbations of the Regge--Wheeler potential. Using an exactly tractable delta-function defect, together with its finite-width Gaussian continuation, Ref.~\cite{OuldElHadj:2026vym} identifies infinite families of exceptional points, derives their large-distance asymptotic organization, and exhibits spectral bridges connecting neighboring Schwarzschild overtones. It also gives a detailed account of the associated square-root splitting, avoided crossings, Riemann-sheet exchange, excitation-factor enhancement, and the resulting time-domain response.

The two works address complementary aspects of the same non-Hermitian spectral structure. Ref.~\cite{OuldElHadj:2026vym} focuses primarily on how localized perturbations generate, organize, and connect exceptional points in the Schwarzschild spectrum. The present work instead focuses on how the resolvent and physical response should be represented through such a coalescence. In particular, we treat the coalescing modes as a single isolated rank-two Riesz cluster and show that its first two contour moments determine a single-valued pair resolvent and, at the exceptional point, the corresponding simple- and double-pole Laurent operators. This formulation avoids resolving or labeling the individual modes and directly yields real-frequency scattering and finite source--observer amplitudes across the exceptional point.

There is also a direct structural connection between the two descriptions. The square-root unfolding of the resonance frequencies found in Ref.~\cite{OuldElHadj:2026vym} is the local two-sheeted manifestation of the same discriminant structure that is encoded in the present formulation by the single-valued cluster invariant $\Delta^{2}$, with the two individual spectral branches recovered only after choosing a local square root. Likewise, Ref.~\cite{OuldElHadj:2026vym} shows that the divergent difference of the two excitation factors is compensated by the vanishing frequency separation, so that their product has a finite and sheet-independent exceptional-point limit. In the present operator formulation, the corresponding regular information is retained directly in the finite double-pole Laurent operator; after projection onto a specified source--observer channel, it determines the coefficient of the term linear in time in the exceptional-point ringdown. Thus the spectral geometry and residue analysis of Ref.~\cite{OuldElHadj:2026vym} and the Riesz--Laurent resolvent construction developed here provide independent and complementary descriptions of the same exceptional-point physics.

\section*{Acknowledgements}

This work was partially supported by Japan Society for the Promotion of Science (JSPS) Grant-in-Aid for Scientific Research Grant Numbers JP25K17402 (O.M.).
This work was supported in part by the COREnet project of RCNP, The University of Osaka (No.\ COREnet065 [O.M.\ and S.O.]).
O.M.\ acknowledges the RIKEN Special Postdoctoral Researcher Program and RIKEN FY2025 Incentive Research Projects.

\appendix

\section{Numerical realization of the contour-moment construction}
\label{app:numerical_realization}

This Appendix specifies a direct route from the analytic quantities in the
main text to a finite-dimensional calculation.  It is not a report of new
numerical data.  The guiding principle is to compute contour moments from
shifted linear systems, rather than to subtract the residues of two nearly
parallel eigenvectors.  The existing complex-scaling implementation of
Ref.~\cite{Ogawa:2026veu} and the Julia package
\texttt{CSMQNM.jl}~\cite{Morikawa:2026CSMQNM} provide the basis and ECS
discretizations, while \texttt{AutoTortoise.jl}~\cite{Morikawa:2026AutoTortoise}
provides the tortoise map and its inverse.  The additional operation needed
here is a quadrature loop around the isolated rank-two cluster.

At a high level, the numerical workflow is as follows.  One first assembles a
single matrix pencil for the complex-scaled operator and uses a preliminary
spectrum only to choose a contour.  Shifted linear systems on that contour
then produce the zeroth and first moments.  These moments determine the
cluster invariants and supply residual tests without requiring individual
eigenvector normalization.  After the EP is located from the zero of
$\Delta^2$, the same shifted solves are contracted with source and observation
vectors to obtain the Laurent and scattering coefficients.  Direct Jost
integration and time evolution are independent validations of the result,
not inputs used to define it.  The following subsections spell out each of
these steps.

\subsection{Tortoise map and continuation along the scaled contour}

For a static metric with
\begin{equation}
    F(r)
    =
    \sum_n a_n r^n,
\end{equation}
\texttt{AutoTortoise.jl} constructs the horizons, selects a static radial
interval, and represents
\begin{equation}
    x(r)
    =
    \int^r\frac{\rmd \bar r}{F(\bar r)}
\end{equation}
by partial fractions.  The same additive constant must be used in the
potential deformation, the source location, and the matching surfaces.  For
Schwarzschild, the numerical inverse should first be checked against
\begin{equation}
    r(x)
    =
    2M
    \left[
        1+W_0\!\left(
            \exp\!\left[
                \frac{x}{2M}-1
            \right]
        \right)
    \right],
    \label{eq:appendix_schwarzschild_inverse}
\end{equation}
where $W_0$ is the principal Lambert function on the physical real interval.

On a global or exterior-scaled contour, write
\begin{equation}
    \zeta
    =
    \zeta_{\theta}(x),
    \qquad
    J(x)
    =
    \frac{\rmd \zeta_{\theta}}{\rmd x}.
\end{equation}
The complex inverse $r(\zeta)$ should be generated by continuation from the
undeformed region.  Starting at a real grid point, one propagates separately
toward the left and right asymptotic segments and uses the preceding value of
$r$ as the initial guess at the next point.  This is more reliable than
starting an independent Newton iteration on the principal logarithmic branch
at every grid point.  Two useful residuals are
\begin{align}
    \epsilon_{x}
    &:=
    \max_j
    \frac{
        \left|x(r_j)-\zeta_j\right|
    }{
        1+|\zeta_j|
    },
    \label{eq:appendix_inverse_residuals_x}
    \\
    \epsilon_{J}
    &:=
    \max_j
    \frac{
        \left|
            \frac{\rmd r_j}{\rmd x}
            -F(r_j)J(x_j)
        \right|
    }{
        1+\left|F(r_j)J(x_j)\right|
    }.
    \label{eq:appendix_inverse_residuals}
\end{align}
They detect a wrong logarithmic branch before it contaminates the potential
matrix.  The contour must also be inspected to ensure that it crosses no
horizon, singularity of the continued potential, or branch cut of the inverse
map.  The Gaussian deformation is evaluated at the same complex coordinate,
\begin{equation}
    \vsub{V}{G}(\zeta)
    =
    \frac{\varepsilon}{M^2}
    \exp\!\left[
        -\frac{(\zeta-d)^2}{2\sigma^2}
    \right].
\end{equation}

\subsection{Choice of discretization and the matrix pencil}

The polynomial--Gaussian implementation of
Ref.~\cite{Ogawa:2026veu} is a convenient first realization of global
scaling.  For real Gaussian ranges its matrix elements require
$\theta<\pi/4$.  The EP frequency in Eq.~\eqref{eq:yang_ep_benchmark} has the
exposure angle
\begin{equation}
    \theta_{\mathrm{exp}}
    =
    \arctan\!\left(
        -\frac{\im\varEP{\omega}}
        {\re\varEP{\omega}}
    \right)
    \simeq
    17.8^{\circ},
\end{equation}
so an initial scan between approximately $25^{\circ}$ and $40^{\circ}$ is
compatible with both conditions, subject to the analyticity of $r(\zeta)$.
The polynomial order and Gaussian ranges should then be varied around the
benchmark choices of Ref.~\cite{Ogawa:2026veu}.

For localized sources, interior observation functionals, and channel
matching, ECS is usually preferable because the interaction region remains
real.  The currently experimental finite-element DVR backend of
\texttt{CSMQNM.jl} represents the
weak form by
\begin{align}
    \left(\mathbf H_{\theta}\right)_{ab}
    &=
    \int\rmd x\,
    \left[
        J^{-1}B_a'B_b'
        +JB_aV(\zeta_{\theta}(x))B_b
    \right],
    \\
    \left(\mathbf N\right)_{ab}
    &=
    \int\rmd x\,
    J B_aB_b,
    \label{eq:appendix_ecs_weak_matrices}
\end{align}
with no complex conjugation in the bilinear products.  The bridge functions
enforce continuity across element boundaries, and the outer endpoint degrees
of freedom impose homogeneous Dirichlet conditions after the outgoing waves
have been damped on the exterior rays.  A finite-difference ECS matrix can
also be used, with $\mathbf N=\mathbf I$, but the conservative flux form is
preferable to a product of first-derivative matrices.

All subsequent steps require a consistent matrix pencil
\begin{equation}
    \mathbf H_{\theta}(q)\mathbf c
    =
    E\mathbf N\mathbf c.
    \label{eq:appendix_matrix_pencil}
\end{equation}
If an overlap-conditioning transformation has been applied, both matrices
must be returned in the same reduced basis.  An eigenvalue-only interface is
insufficient for the present calculation.  In particular, a practical
extension of \texttt{CSMQNM.jl} should expose the conditioned pair
$(\mathbf H_{\theta},\mathbf N)$, together with the basis transformation and
load-vector assembly.  For an orthonormalized real-range calculation one may
simply use the transformed Hamiltonian and $\mathbf N=\mathbf I$.

\subsection{Contour quadrature for the two Riesz moments}

In a nonorthogonal basis, the coefficient-space resolvent corresponding to
Eq.~\eqref{eq:appendix_matrix_pencil} is
\begin{equation}
    \mathbf G_{\theta}(z;q)
    =
    \left[
        z\mathbf N-\mathbf H_{\theta}(q)
    \right]^{-1}
    \mathbf N.
    \label{eq:appendix_pencil_resolvent}
\end{equation}
The factor $\mathbf N$ on the right is essential: the bare inverse of the
pencil maps a Galerkin load vector to coefficients, whereas
$\mathbf G_{\theta}$ represents the resolvent as an operator on coefficient
vectors.

After a preliminary spectrum identifies the pair, choose a circular contour
with center $z_0$ and radius $\rho$ that encloses the two target eigenvalues
but neither a third resonance nor a rotated-continuum pseudostate.  For
$N_{\Gamma}$ midpoint trapezoidal nodes, define
\begin{align}
    \varphi_j
    &=
    \frac{2\pi}{N_{\Gamma}}
    \left(j+\frac{1}{2}\right),
    \\
    z_j
    &=
    z_0+\rho\rme^{\rmi\varphi_j},
    \qquad
    j=0,\ldots,N_{\Gamma}-1.
\end{align}
The two matrix moments are then approximated by
\begin{align}
    \mathbf P_{\Gamma}
    &:=
    \frac{\rho}{N_{\Gamma}}
    \sum_{j=0}^{N_{\Gamma}-1}
    \rme^{\rmi\varphi_j}
    \mathbf G_{\theta}(z_j;q),
    \\
    \boldsymbol{\mathcal M}_{\Gamma}
    &:=
    \frac{\rho}{N_{\Gamma}}
    \sum_{j=0}^{N_{\Gamma}-1}
    z_j\rme^{\rmi\varphi_j}
    \mathbf G_{\theta}(z_j;q).
    \label{eq:appendix_discrete_moments}
\end{align}
Here $\mathbf P_{\Gamma}$ and $\boldsymbol{\mathcal M}_{\Gamma}$ are the matrix
representatives of $\vsub{P}{pair}^{\theta}$ and
$\vsub{\mathcal M}{pair}^{\theta}$, respectively.  At each node, factorize
$z_j\mathbf N-\mathbf H_{\theta}$ once and reuse that factorization for all
right-hand sides.  The nodes are independent and can be distributed across
threads or processes.  For the matrix dimensions used in
Ref.~\cite{Ogawa:2026veu}, forming the full moments is feasible; sparse shifted
solves are preferable for a large FEDVR basis.

When only a small invariant subspace is needed, the full matrices need not be
formed.  Let $\mathbf V_{\mathrm{probe}}$ contain at least two generic probing
vectors and accumulate
\begin{equation}
    \mathbf S_0
    =
    \mathbf P_{\Gamma}\mathbf V_{\mathrm{probe}},
    \qquad
    \mathbf S_1
    =
    \boldsymbol{\mathcal M}_{\Gamma}\mathbf V_{\mathrm{probe}}
\end{equation}
by solving the shifted systems with right-hand side
$\mathbf N\mathbf V_{\mathrm{probe}}$.  If the rank-two truncated singular
value decomposition is
\begin{equation}
    \mathbf S_0
    =
    \mathbf U_2\mathbf\Sigma_2\mathbf V_2^{\dagger},
\end{equation}
then the restricted spectral operator is represented by
\begin{equation}
    \mathbf B_{\mathrm{red}}
    =
    \mathbf U_2^{\dagger}
    \mathbf S_1
    \mathbf V_2
    \mathbf\Sigma_2^{-1}.
    \label{eq:appendix_reduced_cluster_matrix}
\end{equation}
Its trace and traceless part give the same $E_c$ and $\Delta^2$ as the full
moments.  The Hermitian singular value decomposition here is only a stable
numerical device for extracting the range; it does not replace the physical
$c$ product.

For a full-matrix calculation, define
\begin{align}
    E_c
    &=
    \frac{1}{2}\Tr\boldsymbol{\mathcal M}_{\Gamma},
    \\
    \mathbf K_{\Gamma}
    &=
    \boldsymbol{\mathcal M}_{\Gamma}-E_c\mathbf P_{\Gamma},
    \\
    \Delta^2
    &=
    \frac{1}{2}\Tr\left(\mathbf K_{\Gamma}^2\right).
    \label{eq:appendix_discrete_invariants}
\end{align}
Useful internal residuals are
\begin{align}
    \epsilon_P
    &:=
    \frac{
        \left\|\mathbf P_{\Gamma}^2-\mathbf P_{\Gamma}\right\|
    }{
        \left\|\mathbf P_{\Gamma}\right\|
    },
    \label{eq:appendix_moment_residuals_P}
    \\
    \epsilon_{\mathcal M}
    &:=
    \frac{
        \left\|
            \mathbf H_{\theta}\mathbf P_{\Gamma}
            -\mathbf N\boldsymbol{\mathcal M}_{\Gamma}
        \right\|
    }{
        \left\|\mathbf H_{\theta}\right\|
        \left\|\mathbf P_{\Gamma}\right\|
        +
        \left\|\mathbf N\right\|
        \left\|\boldsymbol{\mathcal M}_{\Gamma}\right\|
    },
    \\
    \epsilon_{\mathrm{CH}}
    &:=
    \frac{
        \left\|
            \mathbf K_{\Gamma}^2
            -\Delta^2\mathbf P_{\Gamma}
        \right\|
    }{
        \left\|\mathbf K_{\Gamma}\right\|^2
        +|\Delta^2|
        \left\|\mathbf P_{\Gamma}\right\|
    }.
    \label{eq:appendix_moment_residuals}
\end{align}
One should additionally verify $\Tr\mathbf P_{\Gamma}\simeq2$ and stability
under doubling $N_{\Gamma}$.  The smallest singular value of
$z_j\mathbf N-\mathbf H_{\theta}$ along the contour provides a direct warning
that the contour has been placed too close to an enclosed or excluded
spectral point.

Numerically, this is a contour-integral spectral-projection construction
of the type developed in Refs.~\cite{Beyn:2012,SAKURAI2003119},
specialized here to the first two Riesz moments of the complex-scaled
black-hole pencil.

\subsection{Locating and verifying the exceptional point}

For the model in Sec.~\ref{sec:gaussian_bump_ep}, begin near
\begin{equation}
    q
    =
    (\varepsilon,d),
    \qquad
    \sigma=\frac{1}{\sqrt{2}},
    \qquad
    M=1,
\end{equation}
using the coordinate-converted value in
Eq.~\eqref{eq:converted_ep_center}.  A scaling-angle scan is useful for
identifying an isolated pair and selecting the contour, but the final EP
search need not track either eigenvalue label.  Instead solve the two real
equations
\begin{equation}
    \mathbf F_{\mathrm{EP}}(\varepsilon,d)
    :=
    \begin{pmatrix}
        \re\!\left[\Delta(\varepsilon,d)^2\right]
        \\
        \im\!\left[\Delta(\varepsilon,d)^2\right]
    \end{pmatrix}
    =
    \begin{pmatrix}
        0
        \\
        0
    \end{pmatrix}.
    \label{eq:appendix_ep_root_problem}
\end{equation}
A central finite-difference Jacobian followed by Newton or Broyden iteration
is adequate, provided that the same contour continues to isolate the same
rank-two cluster at every trial point.  This formulation searches for the
zero of a single-valued cluster invariant rather than the crossing of two
branch-labeled QNM frequencies.

A zero of $\Delta^2$ is a repeated eigenvalue in the reduced two-dimensional
problem.  To establish that it is an EP rather than a semisimple degeneracy,
one must also check
\begin{equation}
    \left\|\mathbf K_{\Gamma}\right\|>0,
    \qquad
    \left\|\mathbf K_{\Gamma}^2\right\|\simeq0,
    \qquad
    \rank\mathbf P_{\Gamma}=2.
    \label{eq:appendix_ep_verification}
\end{equation}
Equivalently, at $E=E_c$ the pencil
$\mathbf H_{\theta}-E_c\mathbf N$ should have numerical nullity one within
the isolated root subspace.  Encircling the solution in the
$(\varepsilon,d)$ plane should exchange the two eigenvalue branches, while
$\mathbf P_{\Gamma}$, $E_c$, and $\Delta^2$ return after one circuit.

\subsection{Source amplitudes and real-frequency scattering}

Let $\mathbf b_f$ be the Galerkin load vector of the transformed source and
let $\vsub{\mathbf d}{g}$ be the row that evaluates the chosen observation
functional on a coefficient vector.  At the contour nodes solve
\begin{equation}
    \left(
        z_j\mathbf N-\mathbf H_{\theta}
    \right)
    \mathbf y_j
    =
    \mathbf b_f
\end{equation}
and form the scalar moments
\begin{align}
    p_{g,f}
    &\simeq
    \frac{\rho}{N_{\Gamma}}
    \sum_j
    \rme^{\rmi\varphi_j}
    \vsub{\mathbf d}{g}\mathbf y_j,
    \label{eq:appendix_scalar_projector_moment}
    \\
    m_{g,f}
    &\simeq
    \frac{\rho}{N_{\Gamma}}
    \sum_j
    z_j\rme^{\rmi\varphi_j}
    \vsub{\mathbf d}{g}\mathbf y_j,
    \label{eq:appendix_scalar_first_moment}
    \\
    k_{g,f}
    &=
    m_{g,f}-E_c p_{g,f}.
    \label{eq:appendix_scalar_traceless_moment}
\end{align}
At the EP these give $C_{-1}[g,f]=p_{g,f}$ and
$C_{-2}[g,f]=k_{g,f}$ directly.  The row $\vsub{\mathbf d}{g}$ is assembled
from the physical observation functional; it need not be the Hermitian
conjugate of a right-state coefficient vector when a complex-symmetric
$c$ product is used.  The same shifted factorizations can be reused
for many sources and observers.  If only the pair-projected source is needed,
accumulating the vectors $\mathbf y_j$ before applying
$\vsub{\mathbf d}{g}$ gives the actions of the two Laurent operators without
forming either full matrix.

For a point or narrow observation profile, ECS should place its support in the
undeformed interval.  Direct point evaluation of a globally back-rotated
resonant eigenfunction is exponentially ill-conditioned and is unnecessary
for the contour-moment calculation.  A localized test function provides an
additional regularization and may be narrowed only after basis convergence
has been established.

For scattering at a real frequency, assemble the channel load
$\mathbf b_{\omega}$ for
$U_{\theta}\mathsf V_{\ell}\ket{-\omega}$ and the corresponding observation
row $\vsub{\mathbf d}{\omega}$.  If
\begin{equation}
    v_{\omega}
    :=
    \braket{-\omega|\mathsf V_{\ell}|-\omega},
\end{equation}
then a direct shifted solve at $z_{\omega}=(\omega+\rmi0)^2$ gives
\begin{equation}
    \mathcal T_{\ell}(\omega)
    =
    1+
    \frac{1}{2\rmi\omega}
    \left[
        v_{\omega}
        +
        \vsub{\mathbf d}{\omega}
        \left(
            z_{\omega}\mathbf N-\mathbf H_{\theta}
        \right)^{-1}
        \mathbf b_{\omega}
    \right].
    \label{eq:appendix_direct_transmission}
\end{equation}
The $\rmi0$ prescription can be implemented either by a small positive
imaginary part followed by a stability check or by ECS matching in the
undeformed region.  Applying
Eqs.~\eqref{eq:appendix_scalar_projector_moment} and
\eqref{eq:appendix_scalar_first_moment} to the same channel load and
row gives
\begin{equation}
    \vsub{\tau}{P}
    =
    \frac{p_{\omega}}{2\rmi\omega},
    \qquad
    \vsub{\tau}{K}
    =
    \frac{m_{\omega}-E_c p_{\omega}}{2\rmi\omega},
\end{equation}
which can be inserted into Eq.~\eqref{eq:exact_pair_transmission}.  The
complementary amplitude is most simply obtained from
\begin{equation}
    \mathcal T_{\ell,\mathrm{rest}}
    =
    \mathcal T_{\ell}
    -\mathcal T_{\ell,\mathrm{pair}}.
\end{equation}
This subtraction is between two finite channel amplitudes and is distinct
from subtracting two divergent QNM residues.

An independent real-axis Jost integration is the most stringent benchmark.
With the normalizations of Sec.~\ref{sec:outgoing_resolvent},
\begin{equation}
    W_{\ell}(\omega)
    =
    2\rmi\omega\vsup{A}{in}_{\ell}(\omega),
    \qquad
    \mathcal T_{\ell}(\omega)
    =
    \frac{2\rmi\omega}{W_{\ell}(\omega)}.
\end{equation}
Agreement with Eq.~\eqref{eq:appendix_direct_transmission}, together with
$|\mathcal R_{\ell}|^2+|\mathcal T_{\ell}|^2=1$, tests the potential
continuation, channel normalization, and resolvent extraction independently
of the Riesz decomposition.

\subsection{Convergence and error budget}

The calculation should report stability of the cluster quantities, not only
of the two eigenvalues.  At minimum, one should vary the scaling angle, basis
dimension and ranges, overlap cutoff, integration box, ECS onset and smoothing,
contour center and radius, and $N_{\Gamma}$.  The inverse-map residuals in
Eqs.~\eqref{eq:appendix_inverse_residuals_x} and~\eqref{eq:appendix_inverse_residuals}, the moment residuals in
Eqs.~\eqref{eq:appendix_moment_residuals_P}--\eqref{eq:appendix_moment_residuals}, and the distance of the contour from
the pencil spectrum should be recorded separately.  The final values of
$E_c$, $\Delta^2$, $p_{g,f}$, and $k_{g,f}$ should be insensitive to
$\theta$ within the common convergence window, although the matrix
representatives before back transformation need not be.

Near the EP, the attainable accuracy is set by the larger of the discretization
error and the shifted-solve error.  Iterating the parameter root finder beyond
that floor only produces spurious digits.  If double precision saturates
before the required accuracy is reached, the overlap conditioning, tortoise
inversion, matrix assembly, and shifted solves must all be promoted together;
increasing precision only in the final eigensolver is ineffective.  Finally,
the time-domain coefficients obtained from $C_{-1}$ and $C_{-2}$ should be
checked against a direct inverse Fourier transform or time evolution over a
window in which the continuum discretization is converged.  These tests turn
the analytic identities of the main text into a closed and falsifiable
numerical protocol.

\section{The Aguilar--Balslev--Combes theorem and scaling-angle independence}\label{sec:abc}

In the context of complex scaling, the Aguilar--Balslev--Combes theorem is sometimes quoted only as the rule that
the continuum rotates.  Its real content is a statement about an analytic
family of closed operators and the meromorphic continuation of resolvent
matrix elements~\cite{Aguilar:1971ve,Balslev:1971vb,Simon1972BalslevCombes,Simon:73}.

\subsection{Analytic dilation}

Let the configuration-space dimension be $d$.  For a real dilation parameter
$s$, define the unitary operator
\begin{equation}
  \left[\mathcal U(s)\psi\right](\bm x)
  =\rme^{ds/2}\psi(\rme^s\bm x) .
  \label{eq:real-dilation}
\end{equation}
Complex scaling analytically continues $s$ to $s=\rmi\theta$, with a real
angle $\theta$.  Formally,
\begin{equation}
  \left[\mathcal U(\rmi\theta)\psi\right](\bm x)
  =\rme^{\rmi d\theta/2}
  \psi(\rme^{\rmi\theta}\bm x) .
  \label{eq:complex-dilation}
\end{equation}
For $\theta\ne0$, this transformation is not unitary on the original
$L^2$ space and is generally unbounded.  It should be defined first on a dense set of analytic vectors.  
This is why it is safer to speak of an analytic family of deformed operators than of an everywhere-defined similarity transformation.

For the Hamiltonian
\begin{equation}
  H=T+V(\bm x),
  \qquad
  T=-\frac{\hbar^2}{2m}\nabla^2,
  \label{eq:H-TV}
\end{equation}
the complex-scaled operator is
\begin{equation}
  H_{\theta}
  =\mathcal U(\rmi\theta)H\mathcal U(-\rmi\theta)
  =\rme^{-2\rmi\theta}T
  +V(\rme^{\rmi\theta}\bm x) .
  \label{eq:H-theta}
\end{equation}
Equation~\eqref{eq:H-theta} requires the potential to admit analytic
continuation in the sector swept out by the deformation.

\subsection{Outgoing wave becomes square integrable}
Write the resonance momentum as
\begin{equation}
  \vsub{k}{R}=|\vsub{k}{R}|\rme^{-\rmi\varphi},
  \qquad 0<\varphi<\frac{\pi}{4} .
  \label{eq:k-polar}
\end{equation}
Equivalently, the resonance energy
$E_R=\hbar^2 k_R^2/(2m)$ has argument $-2\varphi$; hence the
positive-energy resonance sector $-\pi/2<\arg E_R<0$ corresponds to
$0<\varphi<\pi/4$ in the momentum plane.

On the positive half-line, the scaled outgoing wave is
\begin{align}
  \rme^{\rmi \vsub{k}{R}r\rme^{\rmi\theta}}
  &=\rme^{
  \rmi|\vsub{k}{R}|r\cos(\theta-\varphi)
  }
  \notag\\
  &\quad\times
  \rme^{-|\vsub{k}{R}|r
  \sin(\theta-\varphi)} .
  \label{eq:scaled-outgoing}
\end{align}
It decays for $\theta>\varphi$.  
The same condition regularizes the outgoing wave on the negative half-line when the contour is rotated consistently.
Thus the Siegert state becomes an $L^2$ eigenfunction of $H_\theta$ once the rotated continuum has passed below the pole.  
The divergent unscaled wave and the normalizable scaled wave are analytic continuations of one solution, 
not different physical states.

\subsection{The Aguilar--Balslev--Combes theorem}
The main statement is as follows:
\begin{theorem}[Aguilar--Balslev--Combes, type-A one-threshold form]
Let $H_0=-\hbar^2\nabla^2/(2m)$ on $L^2(\bbR^d)$ with domain $H^2(\bbR^d)$,
and let $H=H_0+V$ be self-adjoint.  Suppose there is $\theta_0>0$ such that
for $|\im s|<\theta_0$ the dilation
$V_s=\mathcal U(s)V\mathcal U(-s)$ is analytic as an operator from the domain $D(H_0)$ with its graph norm
to $L^2$, and $V_s(H_0+1)^{-1}$ is compact locally uniformly in $s$.  Assume
also that $H(s)=\rme^{-2s}H_0+V_s$ is a closed type-A family on the common
domain $D(H_0)$.

For $0<\theta<\theta_0$, put $H_\theta=H(\rmi\theta)$.  Then
\begin{enumerate}[label=(\roman*)]
  \item the essential spectrum of $H_\theta$ is given by $\rme^{-2\rmi\theta}[0,\infty)$;
  \item the remaining spectrum away from this ray consists of isolated
  eigenvalues of finite algebraic multiplicity;
  \item matrix elements $\braket{{\phi},{(z-H)^{-1}\psi}}$ ($\phi$, $\psi\in\mathcal A$, where $\mathcal A$ is a dense set of dilation-analytic vectors) between
  dilation-analytic vectors continue meromorphically from the physical
  resolvent set into the wedge swept by the ray, and their poles can occur only at discrete eigenvalues of $H_\theta$ in this wedge; conversely, each such discrete eigenvalue is detected as a pole for suitable dilation-analytic vectors; and
  \item these eigenvalues and their algebraic multiplicities are independent
  of $\theta$ while they remain in a common uncovered wedge.  Negative
  isolated eigenvalues of $H$ remain unchanged.
\end{enumerate}
\label{thm:abc-type-A}
\end{theorem}

For the black-hole problem considered in the main text, the starting
operator on the real $x$ axis is the Schr\"odinger operator with a real
radial potential, while the QNM condition is encoded in the meromorphic
continuation of its outgoing resolvent.  In the spectral variable
$z=\omega^2$, complex scaling rotates the continuous spectrum by
$-2\theta$, while an exposed resonance eigenvalue $E_n=\omega_n^2$
remains independent of $\theta$ as long as it stays in a common
uncovered wedge and the deformation remains within a common analyticity
domain.  This is the sense in which the $\theta$ independence holds in the exact theory.  The residual $\theta$ dependence
observed after basis truncation is therefore a numerical convergence
diagnostic.  Exterior complex scaling implements the same spectral
deformation principle while leaving a finite interior region undeformed.

\section{Derivation of the individual eigenvalues and spectral projectors}
\label{app:individual_projectors}

In this appendix we give a detailed derivation of the individual
eigenvalues and spectral projectors associated with the isolated
rank-two Riesz cluster away from the exceptional point.
The purpose is to make explicit how Eq.~\eqref{eq:individual_projectors_from_cluster}
follows from the cluster quantities
$\vsub{P}{pair}^{\theta}$,
$\vsub{K}{pair}^{\theta}$,
$\vsub{E}{c}$, and $\Delta^2$,
without introducing separately normalized left and right eigenvectors.

For notational simplicity, we suppress the parameter dependence on $q$
and write
\begin{align}
    P &:= \vsub{P}{pair}^{\theta},
    \\
    M &:= \vsub{M}{pair}^{\theta},
    \\
    E_c &:= \vsub{E}{c},
    \\
    K &:= \vsub{K}{pair}^{\theta}.
\end{align}
The relevant two-dimensional invariant root subspace is
\begin{equation}
    \mathcal{X} := \operatorname{im} P.
\end{equation}
Since $P$ is the Riesz projector onto this subspace,
\begin{equation}
    P|_{\mathcal{X}} = I_{\mathcal{X}},
\end{equation}
where $I_{\mathcal{X}}$ denotes the identity operator on
$\mathcal{X}$.  Moreover,
\begin{equation}
    M = HP,
\end{equation}
so that the restriction of $M$ to $\mathcal{X}$ coincides with the
restriction of $H$,
\begin{equation}
    M|_{\mathcal{X}} = H|_{\mathcal{X}}.
\end{equation}
By definition,
\begin{equation}
    K = M - E_c P,
\end{equation}
and therefore
\begin{equation}
    K|_{\mathcal{X}}
    =
    H|_{\mathcal{X}}
    -
    E_c I_{\mathcal{X}}.
    \label{eq:app_K_restricted}
\end{equation}

The spectral center is
\begin{equation}
    E_c
    =
    \frac{1}{2}\Tr M.
\end{equation}
Since $\dim\mathcal{X}=2$, Eq.~\eqref{eq:app_K_restricted} implies
\begin{align}
    \Tr_{\mathcal{X}}
    \left(
        K|_{\mathcal{X}}
    \right)
    &=
    \Tr_{\mathcal{X}}
    \left(
        H|_{\mathcal{X}}
    \right)
    -
    E_c
    \Tr_{\mathcal{X}} I_{\mathcal{X}}
    \\
    &=
    2E_c - 2E_c
    \\
    &=0.
    \label{eq:app_K_traceless}
\end{align}
Thus $K|_{\mathcal{X}}$ is a traceless operator on a
two-dimensional space.

For a two-dimensional operator $B$, the Cayley--Hamilton identity is
\begin{equation}
    B^2
    -
    (\Tr B)B
    +
    (\det B)I
    =
    0.
\end{equation}
Applying this identity to $K|_{\mathcal{X}}$ and using
Eq.~\eqref{eq:app_K_traceless}, one obtains
\begin{equation}
    \left(
        K|_{\mathcal{X}}
    \right)^2
    =
    -
    \det
    \left(
        K|_{\mathcal{X}}
    \right)
    I_{\mathcal{X}}.
\end{equation}
For a two-dimensional operator,
\begin{equation}
    \det B
    =
    \frac{1}{2}
    \left[
        (\Tr B)^2
        -
        \Tr(B^2)
    \right].
\end{equation}
Using again $\Tr_{\mathcal{X}}K=0$ and the definition
\begin{equation}
    \Delta^2
    :=
    \frac{1}{2}\Tr K^2,
\end{equation}
we obtain
\begin{equation}
    K^2
    =
    \Delta^2 P.
    \label{eq:app_K_squared}
\end{equation}
Here and below the identity is understood as an operator identity
embedded in the full space; on $\mathcal{X}$, the projector $P$ acts
as the identity.

\subsection{Eigenvalues of the restricted operator}

Away from the exceptional point,
\begin{equation}
    \Delta^2 \neq 0.
\end{equation}
Choose either local branch of the square root,
\begin{equation}
    \Delta = \sqrt{\Delta^2}.
\end{equation}
Equation~\eqref{eq:app_K_squared} shows that every eigenvalue $\kappa$
of $K|_{\mathcal{X}}$ satisfies
\begin{equation}
    \kappa^2 = \Delta^2.
\end{equation}
Since $K|_{\mathcal{X}}$ is traceless and the subspace is
two-dimensional, its two eigenvalues are therefore
\begin{equation}
    \kappa_{\pm} = \pm \Delta.
\end{equation}

Let $v_{\pm}\in\mathcal{X}$ denote the corresponding eigenvectors,
\begin{equation}
    Kv_{\pm}
    =
    \pm\Delta v_{\pm}.
    \label{eq:app_K_eigenvectors}
\end{equation}
On $\mathcal{X}$ we have
\begin{equation}
    H
    =
    E_c I_{\mathcal{X}}
    +
    K.
\end{equation}
Hence
\begin{align}
    Hv_{\pm}
    &=
    \left(
        E_c I_{\mathcal{X}} + K
    \right)
    v_{\pm}
    \\
    &=
    \left(
        E_c \pm \Delta
    \right)
    v_{\pm}.
\end{align}
The two eigenvalues of the restricted spectral operator are thus
\begin{equation}
    E_{\pm}
    =
    E_c \pm \Delta.
    \label{eq:app_Epm}
\end{equation}

\subsection{Construction of the individual spectral projectors}

Let $P_{+}$ and $P_{-}$ denote the spectral projectors onto the
$E_{+}$ and $E_{-}$ eigenspaces, respectively.  Since the two
eigenspaces exhaust the isolated rank-two cluster,
\begin{equation}
    P_{+}+P_{-}=P.
    \label{eq:app_projector_sum}
\end{equation}
Furthermore, $K$ acts as $+\Delta$ on the range of $P_{+}$ and as
$-\Delta$ on the range of $P_{-}$.  Therefore
\begin{equation}
    K
    =
    \Delta P_{+}
    -
    \Delta P_{-},
\end{equation}
or equivalently,
\begin{equation}
    P_{+}-P_{-}
    =
    \frac{K}{\Delta}.
    \label{eq:app_projector_difference}
\end{equation}
Solving Eqs.~\eqref{eq:app_projector_sum} and
\eqref{eq:app_projector_difference} gives
\begin{equation}
    P_{\pm}
    =
    \frac{1}{2}
    \left(
        P
        \pm
        \frac{K}{\Delta}
    \right).
    \label{eq:app_Ppm}
\end{equation}
Restoring the full notation,
\begin{align}
    E_{\pm}(q)
    &=
    \vsub{E}{c}(q)
    \pm
    \Delta(q),
    \\
    P_{\pm}^{\theta}(q)
    &=
    \frac{1}{2}
    \left[
        \vsub{P}{pair}^{\theta}(q)
        \pm
        \frac{
            \vsub{K}{pair}^{\theta}(q)
        }{
            \Delta(q)
        }
    \right].
\end{align}

It is useful to verify directly that
Eq.~\eqref{eq:app_Ppm} indeed defines spectral projectors.
The Riesz identities imply
\begin{equation}
    P^2=P,
    \qquad
    PK=KP=K.
    \label{eq:app_PK_identities}
\end{equation}
Using Eqs.~\eqref{eq:app_K_squared} and
\eqref{eq:app_PK_identities}, we find
\begin{align}
    P_{\pm}^2
    &=
    \frac{1}{4}
    \left(
        P
        \pm
        \frac{K}{\Delta}
    \right)^2
    \\
    &=
    \frac{1}{4}
    \left(
        P^2
        \pm
        \frac{2PK}{\Delta}
        +
        \frac{K^2}{\Delta^2}
    \right)
    \\
    &=
    \frac{1}{4}
    \left(
        2P
        \pm
        \frac{2K}{\Delta}
    \right)
    \\
    &=
    P_{\pm}.
\end{align}
Similarly,
\begin{align}
    P_{+}P_{-}
    &=
    \frac{1}{4}
    \left(
        P+\frac{K}{\Delta}
    \right)
    \left(
        P-\frac{K}{\Delta}
    \right)
    \\
    &=
    \frac{1}{4}
    \left(
        P-\frac{K^2}{\Delta^2}
    \right)
    \\
    &=0.
\end{align}
Thus the two projectors are idempotent and mutually annihilating.
They are spectral projectors of a non-Hermitian operator and need not
be orthogonal projectors with respect to a Hermitian inner product.

One may also verify the spectral relation directly.  From
Eq.~\eqref{eq:app_K_squared},
\begin{align}
    KP_{\pm}
    &=
    \frac{1}{2}
    \left(
        K
        \pm
        \frac{K^2}{\Delta}
    \right)
    \\
    &=
    \frac{1}{2}
    \left(
        K
        \pm
        \Delta P
    \right)
    \\
    &=
    \pm\Delta P_{\pm}.
    \label{eq:app_KPpm}
\end{align}
Since
\begin{equation}
    HP=E_cP+K,
\end{equation}
Eqs.~\eqref{eq:app_PK_identities} and
\eqref{eq:app_KPpm} give
\begin{align}
    HP_{\pm}
    &=
    \left(
        E_cP+K
    \right)
    P_{\pm}
    \\
    &=
    E_cP_{\pm}
    \pm
    \Delta P_{\pm}
    \\
    &=
    E_{\pm}P_{\pm}.
\end{align}
This confirms directly that $P_{\pm}$ are the spectral projectors
associated with $E_{\pm}$.

\subsection{Equivalent derivation from the pair resolvent}

The same result follows immediately by partial-fraction decomposition
of the exact rank-two pair resolvent.  The cluster formula is
\begin{equation}
    R_{\mathrm{pair}}^{\theta}(z)
    =
    \frac{
        (z-E_c)P+K
    }{
        (z-E_c)^2-\Delta^2
    }.
    \label{eq:app_pair_resolvent}
\end{equation}
Introduce
\begin{equation}
    a:=z-E_c.
\end{equation}
Then
\begin{equation}
    R_{\mathrm{pair}}^{\theta}(z)
    =
    \frac{
        aP+K
    }{
        (a-\Delta)(a+\Delta)
    }.
\end{equation}
The poles therefore occur at
\begin{equation}
    a=\pm\Delta,
\end{equation}
or equivalently at
\begin{equation}
    z=E_c\pm\Delta=E_{\pm}.
\end{equation}

To determine the residues, write
\begin{equation}
    \frac{
        aP+K
    }{
        (a-\Delta)(a+\Delta)
    }
    =
    \frac{A}{a-\Delta}
    +
    \frac{B}{a+\Delta}.
\end{equation}
Multiplication by the denominator gives
\begin{equation}
    aP+K
    =
    A(a+\Delta)
    +
    B(a-\Delta).
\end{equation}
Equating the coefficient of $a$ and the constant term yields
\begin{align}
    A+B &= P,
    \\
    \Delta(A-B) &= K.
\end{align}
Hence
\begin{align}
    A
    &=
    \frac{1}{2}
    \left(
        P+\frac{K}{\Delta}
    \right),
    \\
    B
    &=
    \frac{1}{2}
    \left(
        P-\frac{K}{\Delta}
    \right).
\end{align}
Therefore
\begin{equation}
    R_{\mathrm{pair}}^{\theta}(z)
    =
    \frac{P_{+}}{z-E_{+}}
    +
    \frac{P_{-}}{z-E_{-}},
    \label{eq:app_resolvent_simple_poles}
\end{equation}
with
\begin{equation}
    P_{\pm}
    =
    \frac{1}{2}
    \left(
        P
        \pm
        \frac{K}{\Delta}
    \right).
\end{equation}
Since the residue of the resolvent at an isolated simple eigenvalue is
the corresponding Riesz spectral projector,
Eq.~\eqref{eq:app_resolvent_simple_poles} gives the same
$P_{\pm}$ as the preceding algebraic construction.

\subsection{Choice of square-root branch and exchange of mode labels}

The cluster construction fundamentally determines $\Delta^2$, not a
globally single-valued square root $\Delta$.  Away from the exceptional
point, one may choose either local branch,
\begin{equation}
    \Delta
    =
    \sqrt{\Delta^2}.
\end{equation}
Changing the branch sends
\begin{equation}
    \Delta
    \longmapsto
    -\Delta.
\end{equation}
Equations~\eqref{eq:app_Epm} and \eqref{eq:app_Ppm} then imply
\begin{equation}
    E_{+}
    \longleftrightarrow
    E_{-},
    \qquad
    P_{+}
    \longleftrightarrow
    P_{-}.
\end{equation}
Thus a branch change merely exchanges the labels assigned to the two
individual modes.  By contrast, the cluster quantities
\begin{equation}
    P,
    \qquad
    K,
    \qquad
    E_c,
    \qquad
    \Delta^2
\end{equation}
are unchanged.  This is the sense in which the latter provide
single-valued variables for the isolated rank-two spectral cluster.

\subsection{Exceptional-point limit}

The individual-projector representation also makes explicit why
separately labeled modal projectors become ill-conditioned near the
exceptional point.  From Eq.~\eqref{eq:app_Ppm},
\begin{equation}
    P_{\pm}
    =
    \frac{1}{2}P
    \pm
    \frac{K}{2\Delta}.
    \label{eq:app_projector_singularity}
\end{equation}
As the exceptional point is approached,
\begin{equation}
    \Delta^2\longrightarrow 0.
\end{equation}
At a second-order exceptional point,
\begin{equation}
    K\neq0,
    \qquad
    K^2=0.
\end{equation}
Consequently, the two individual projectors may contain opposite
terms proportional to $1/\Delta$.  Their sum nevertheless remains
finite,
\begin{equation}
    P_{+}+P_{-}=P.
\end{equation}
The apparent singularity is therefore associated with the
decomposition into separately labeled modes rather than with the
rank-two invariant subspace itself.

The same cancellation is transparent at the level of the resolvent.
Away from the exceptional point,
\begin{equation}
    \vsub{R}{pair}^{\theta}(z)
    =
    \frac{P_{+}}{z-E_{+}}
    +
    \frac{P_{-}}{z-E_{-}},
\end{equation}
where the two terms may separately become ill-conditioned as
$\Delta\to0$.  Their sum, however, is exactly the branch-independent
expression in Eq.~\eqref{eq:app_pair_resolvent}.  Taking the
exceptional-point limit gives
\begin{equation}
    \vsub{R}{pair}^{\theta}(z)
    \longrightarrow
    \frac{\varEP{K}}{
        (z-\varEP{E})^2
    }
    +
    \frac{\varEP{P}}{
        z-\varEP{E}
    },
    \label{eq:app_EP_resolvent}
\end{equation}
where
\begin{equation}
    \varEP{E}=E_c(\varEP{q}),
\end{equation}
and
\begin{equation}
    \varEP{K}
    :=
    K(\varEP{q}),
    \qquad
    \varEP{P}
    :=
    P(\varEP{q}).
\end{equation}
Thus the two simple poles reorganize into the simple- plus
double-pole Laurent form of a second-order Jordan block.  The
divergence of the individual projectors in
Eq.~\eqref{eq:app_projector_singularity} is therefore a singularity of
the modal coordinates, whereas the cluster projector $P$, the
nilpotent operator $K$, and the exact pair resolvent remain finite
objects through the coalescence.

\bibliographystyle{utphys}
\bibliography{ref}
\end{document}